\documentclass[format=acmsmall, review=false]{acmart}
\usepackage{acm-arxiv}
\usepackage{booktabs} 
\usepackage{natbib}
\setcitestyle{authoryear}

\usepackage[multiple]{footmisc} 

\usepackage{comment}

\usepackage{xcolor}
\usepackage{color}
\usepackage[colorinlistoftodos,textsize=tiny]{todonotes}
\newcommand{\Comments}{1}
\definecolor{gray}{gray}{0.5}
\definecolor{lightred}{rgb}{1,0.6,0.6}
\definecolor{darkgreen}{rgb}{0,0.5,0}
\definecolor{myorange}{rgb}{0.8,0.7,0.5}
\definecolor{darkblue}{rgb}{0.0,0.0,0.2}

\newcommand{\mynote}[2]{\ifnum\Comments=0\textcolor{#1}{#2}\fi}
\newcommand{\mytodo}[2]{\ifnum\Comments=0%
  \todo[linecolor=#1!80!black,backgroundcolor=#1,bordercolor=#1!80!black]{#2}\fi}

\newcommand{\maneesha}[1]{{\mynote{myorange}{[MP: #1]}}}

\newcommand{\btw}[1]{}
\ifnum\Comments=1               
\fi

\hypersetup{colorlinks=true,breaklinks=true,bookmarks=true,urlcolor=darkblue,citecolor=darkblue,linkcolor=darkblue,bookmarksopen=false,draft=false}

\newtheorem{theorem}{Theorem}
\newtheorem{claim}{Claim}
\newtheorem{lemma}{Lemma}
\newtheorem{fact}{Fact}
\newtheorem{proposition}{Proposition}
\newtheorem{corollary}{Corollary}

\newtheorem{observation}{Observation}
\newtheorem{construction}{Construction}

\theoremstyle{definition}\newtheorem{definition}{Definition}
\theoremstyle{definition}\newtheorem{axiom}{Axiom}

\usepackage{dsfont}
\usepackage{bm}
\newcommand{\ones}{\mathds{1}}  

\newcommand{\reals}{\mathbb{R}}
\newcommand{\A}{\mathcal{A}}

\newcommand{\R}{\mathcal{R}}

\newcommand{\Y}{\mathcal{Y}}

\renewcommand{\bar}[1]{\overline{#1}}

\newcommand{\valtrades}{\mathsf{ValTrades}}
\newcommand{\valhist}{\mathsf{ValHist}}
\newcommand{\hsum}{\mathsf{sum}}

\newcommand{\fee}{\mathsf{fee}}

\usepackage{xspace}

\newcommand{\nodominatedtrades}{\textsc{NoDominatedTrades}\xspace}
\newcommand{\pathindependence}{\textsc{PathIndependence}\xspace}
\newcommand{\strongpathindependence}{\textsc{StrongPathIndependence}\xspace}
\newcommand{\liquidation}{\textsc{Liquidation}\xspace}

\newcommand{\demandresponsiveness}{\textsc{DemandResponsiveness}\xspace}
\newcommand{\boundedreserves}{\textsc{BoundedReserves}\xspace}
\newcommand{\worstcaseloss}{\textsc{WorstCaseLoss}\xspace}

\renewcommand{\vec}[1]{{\mathbf{#1}}}
\renewcommand{\r}{\vec{r}}
\newcommand{\p}{\vec{p}}
\newcommand{\q}{\vec{q}}
\newcommand{\qZ}{\vec{q}_0}
\newcommand{\qh}{\vec{q}_h}

\newcommand{\0}{\vec{0}}

\title{An Axiomatic Characterization of CFMMs and Equivalence to Prediction Markets}
\author{Rafael Frongillo, Maneesha Papireddygari and Bo Waggoner}

\begin{abstract}
  Constant-function market makers (CFMMs), such as Uniswap, are automated exchanges offering trades among a set of assets.
  We study their technical relationship to another class of automated market makers, cost-function prediction markets.
  We first introduce axioms for market makers and show that CFMMs with concave potential functions characterize ``good'' market makers according to these axioms.
  We then show that every such CFMM on $n$ assets is equivalent to a cost-function prediction market for events with $n$ outcomes.
  Our construction directly converts a CFMM into a prediction market and vice versa.
  
  Conceptually, our results show that desirable market-making axioms are equivalent to desirable information-elicitation axioms, i.e., markets are good at facilitating trade if and only if they are good at revealing beliefs.
  For example, we show that every CFMM implicitly defines a \emph{proper scoring rule} for eliciting beliefs; the scoring rule for Uniswap is unusual, but known. 
  From a technical standpoint, our results show how tools for prediction markets and CFMMs can interoperate.
  We illustrate this interoperability by showing how liquidity strategies from both literatures transfer to the other, yielding new market designs.
\end{abstract}
\begin{document}

\maketitle

\break

\section{Introduction}
\label{sec:introduction}
\maneesha{H}

A prediction market is designed to elicit forecasts, such as for elections or sporting events.
Participants trade shares in $n$ securities, each tied to one of the $n$ possible outcomes, paying out \$1 if that outcome occurs;
the prices form a probability distribution over the outcomes representing an aggregate belief.
To resolve thin market problems, \citet{hanson2003combinatorial} and \citet{chen2007utility} propose introducing \emph{automated market makers}, mechanisms that offers to buy or sell any number of securities, with prices determined by the past history of transactions.
A significant amount of research has investigated the design of such markets based on convex \emph{cost functions} $C$ \cite{abernethy2013efficient,frongillo2018axiomatic}.
To acquire a bundle of securities $\r\in \reals^n$, the trader pays $C(\q+\r) -C(\q)$ cash to the market maker, where $\q \in \reals^n$ is the total numbers of securities sold so far.
A prominent example is the log market scoring rule (LMSR) due to \citet{hanson2003combinatorial}, given by $C(\q) = b \log \sum_{i=1}^n e^{q_i/b}$ for a liquidity parameter $b>0$.

More recently, inspired by blockchain applications, there has been significant theoretical and practical interest in the design of decentralized financial markets \cite{capponi2021adoption,angeris2022constant}.
These markets allow trade between $n$ assets without a fixed unit of exchange (``cash'').
Interestingly, like prediction markets, many decentralized exchanges also employ automated market makers, but for different reasons: automated market makers tend to have lower on-chain implementation costs than traditional order books~\citep[\S 4.1]{jensen2021introduction}.
The dominant paradigm is the \emph{constant-function market maker (CFMM)} based on a concave \emph{potential function} $\varphi$.
Here, if $\q \in \reals^n$ is the amount of market maker holdings or ``reserves'' in the $n$ assets, then a trade $\r\in\reals^n$ is accepted if $\varphi(\q + \r) = \varphi(\q)$, i.e., it keeps the potential function constant.
One of the most popular CFMMs is Uniswap where $\varphi(\q)= \sqrt{q_1 q_2}$.

Despite clear similarities between these two settings, prior work has considered prediction markets to be only analogous to CFMMs, but not technically related.%
\footnote{To give two examples, \citet[\S 3.2]{angeris2020improved} argue that the two market makers can diverge, and \citet{bichuch2022axioms} state the following: ``In a separate context, [automated market makers] for prediction markets were first proposed in \citep{hanson2003combinatorial}. Such a structure is fundamentally different from the constant function market makers considered herein insofar as a prediction market includes a terminal time at which bets are realized.'' }
Indeed, there are obvious differences between the settings: prediction markets rely on ``cash'', while CFMMs do not require a special unit of account; prediction markets only offer certain types of outcome-dependent securities (and can manufacture as many of these as they want), while CFMMs deal in arbitrary assets in limited supply; and so on.
Furthermore, the \emph{goals} of the two types of markets are very different: prediction markets seek to elicit information, whereas CFMMs seek to facilitate trade.

Nonetheless, we show a tight technical equivalence between prediction markets and CFMMs.
We give straightforward reductions to transform a cost function $C$ for a prediction market into a potential function $\varphi$ for a CFMM, and vice versa.
The reductions give one-to-one correspondence between trades placed in both markets, while preserving the relevant market properties.
An important implication is that markets designed to elicit information are also good at facilitating trade, and vice versa.
We illustrate our results with several examples; among them,
the prediction market equivalent to Uniswap, a new homogeneous CFMM from LMSR (Figure~\ref{fig:intro-figure-lmsr}), and a close connection between Brier score and a hybrid CFMM.
We hope our results can inspire new market designs and insights in both literatures.

\begin{figure}[t]
  \includegraphics[width = 1.8in]{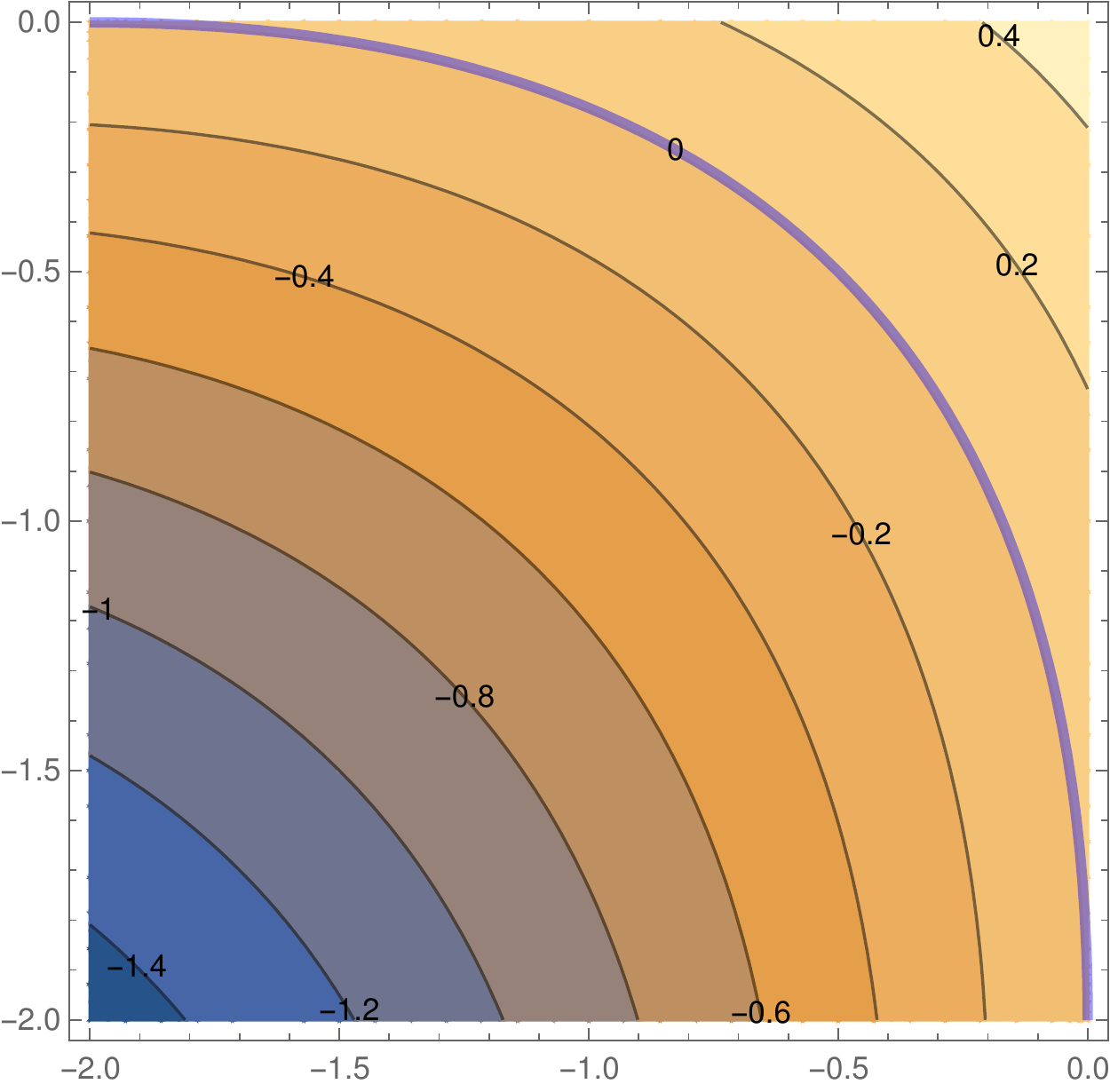}
  \includegraphics[width = 1.8in]{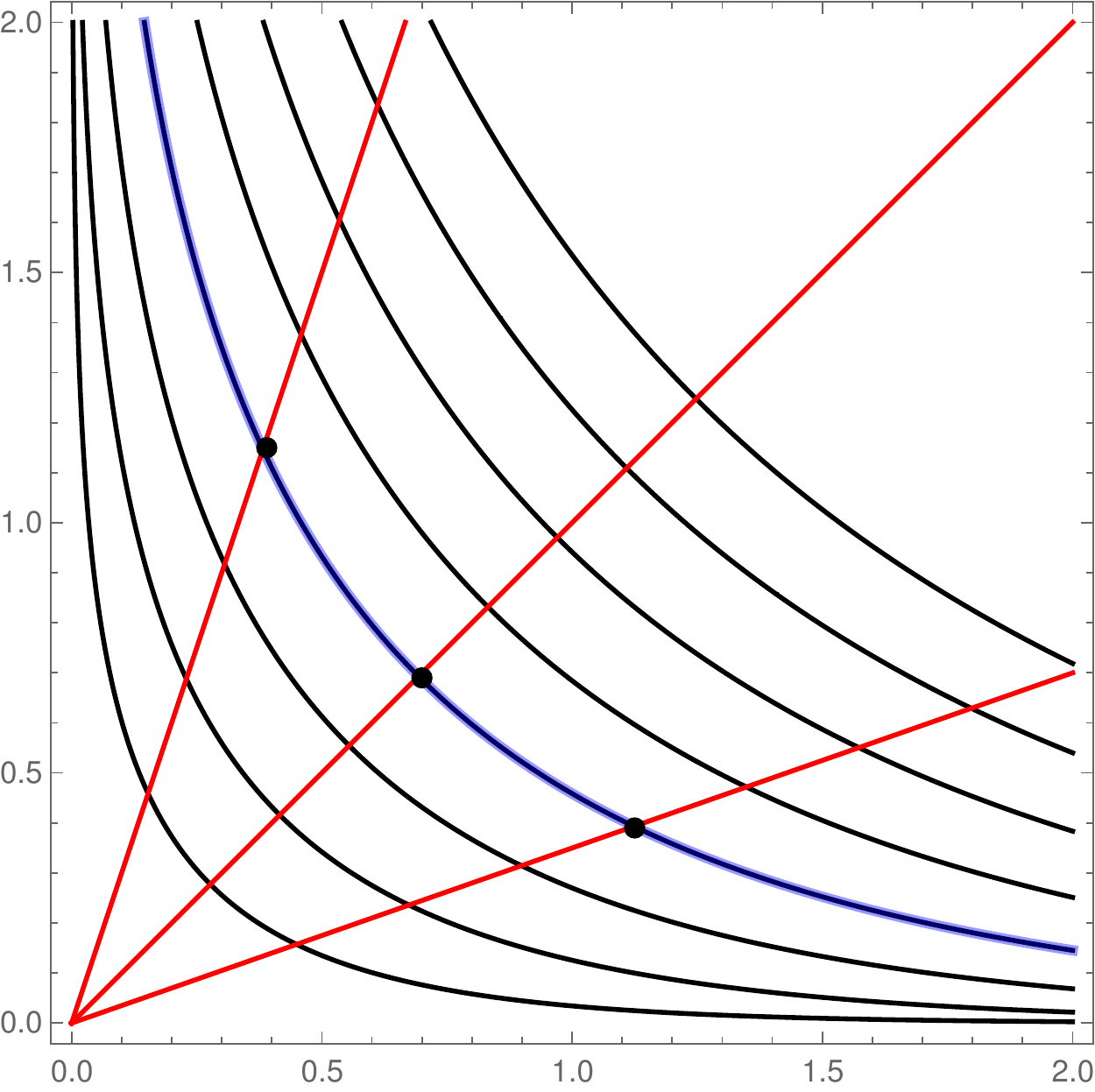}
  \includegraphics[width = 1.8in,trim={5 5 20 5},clip]{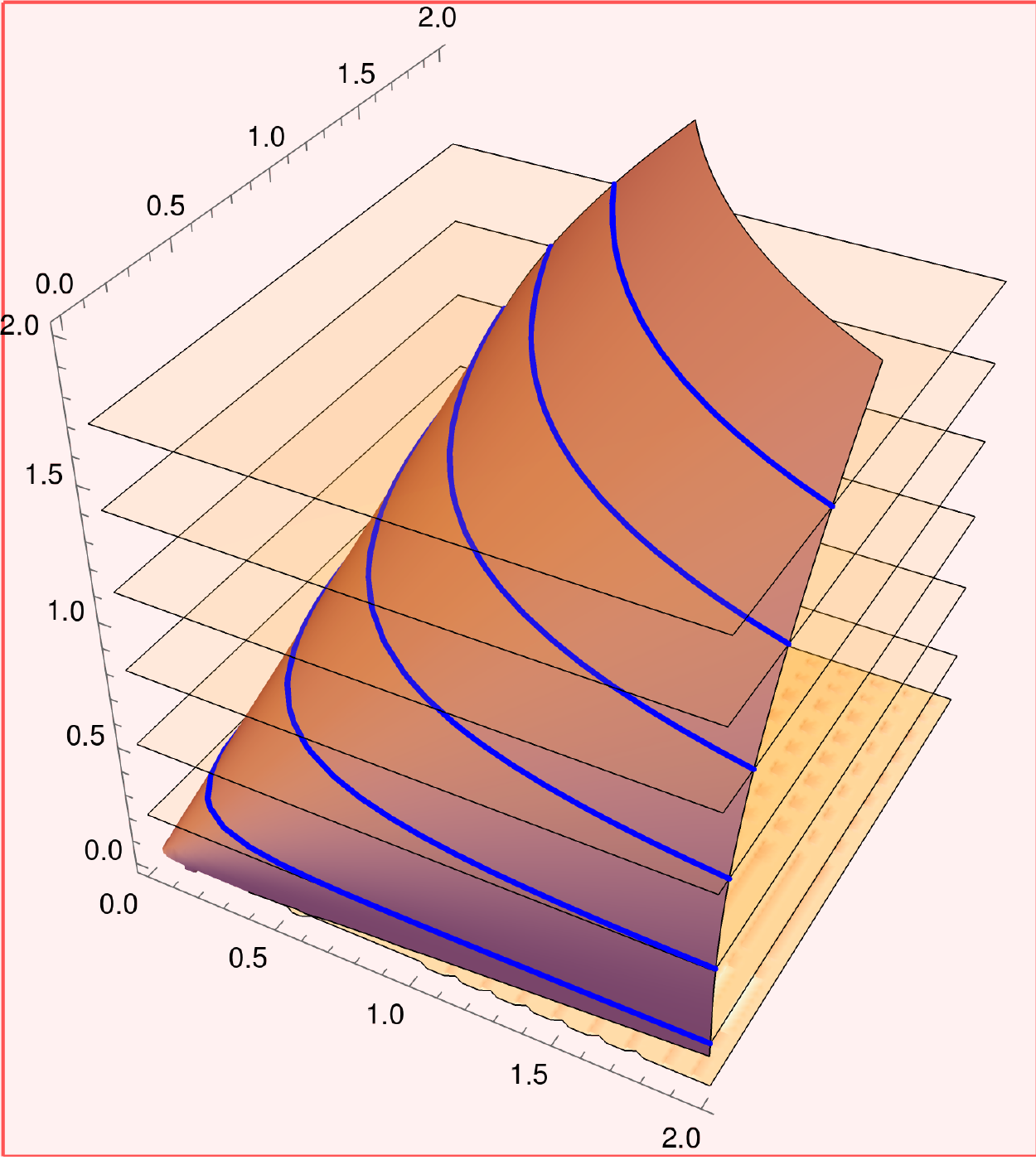}
  \caption{The ``reserves-aware'' Construction~\ref{cons:bounded-reserves} applied to the log market scoring rule (LMSR) of \citet{hanson2003combinatorial}.  (L) The level sets of LMSR, with the 0 level set highlighted in blue.  (M) Taking the 0 level set of LMSR and reflecting it into the positive orthant, we assign $\varphi=1$ on this set (again in blue).  We then derive the other level sets by shrinking or expanding uniformly toward or away from the origin: the $\alpha$ level set, $\alpha > 0$, is given by scaling the 1 level set by $\alpha$.  As we will see, this value $\alpha$ can be interpreted as the liquidity level of the market.  (R) The resulting 1-homogeneous potential $\varphi$, which satisfies all of our axioms (\S~\ref{sec:bounded-res-examples}).}
  \label{fig:intro-figure-lmsr}
\end{figure}

\subsection{Outline of results}
We begin in \S~\ref{sec:characterization} with an axiomatic characterization of the class of constant-function market makers (CFMMs) of interest in this paper.
Focusing on ``vanilla'' CFMMs that do not adapt liquidity, we show in  Proposition~\ref{prop:fixedq-cfmm} that an automated market maker satisfies a certain set of baseline axioms if and only if it can be represented by a CFMM with a concave, increasing $\varphi$.
This result assumes the initial reserves are fixed; we give a similar result in Theorem~\ref{thm:itsacfmm} allowing any initial reserves where $\varphi$ is quasiconcave and continuous.

We give our main equivalence in \S~\ref{sec:equivalence-cfmms-pred}.
We first recall that cost-function prediction markets, where $C$ satisfies certain criteria, have been characterized by prior work as automated market makers satisfying a path independence axiom and an incentive-compatibility axiom (see Fact~\ref{fact:elicitation-cost-function}).
We give a reduction from any such $C$ to a concave, increasing $\varphi$, resulting in a CFMM that satisfies desirable market-making axioms (Theorem \ref{thm:cost-to-CFMM}).
Then, we give a reduction from any concave, increasing $\varphi$ to a convex $C$ satisfying the criteria, resulting in a prediction market with desirable incentive-compatibility axioms (Theorem \ref{thm:CFMM-to-cost}).
The conceptual implication is that design for good information-elicitation properties results in good market-making properties, and vice versa.
In fact, an immediate corollary is that any CFMM defines a \emph{proper scoring rule}, a function for eliciting truthful predictions (Corollary \ref{cor:CFMM-scoring}).

Our equivalence above is not fully satisfying in some cases, however, because for some initial conditions the resulting CFMMs can run out of their reserves.
In \S~\ref{sec:reserves-and-liqudity}, we begin by characterizing markets that keep their reserves nonnegative in addition to our previous axioms.
We then give a more sophisticated reserves-aware reduction from prediction markets to CFMMs, using the \emph{perspective transform} technique from convex analysis.
The result is a CFMM $\varphi$ whose level sets are not parallel, as are those of cost functions (Figure~\ref{fig:intro-figure-lmsr}(L)), but instead ``stretched'', with curvature adapting to the amount of reserves available (Figure~\ref{fig:intro-figure-lmsr}(M)).
We show that every 1-homogeneous (scale-symmetric) potential function, the dominant CFMM design in practice, can be obtained via our construction.

We conclude in \S~\ref{sec:future} with a discussion of other practical considerations, mainly \emph{liquidity adaptation}.
In applications, CFMMs are designed to allow the reserves to grow (e.g.\ by lending), and the liquidity provided should grow as well \citep{angeris2020when}.
Similarly, prediction-market research has studied how to automatically grow and adapt the liquidity of the market over time, again using the perspective transform \citep{abernethy2014general}.
We begin an investigation of how these lines of research can inform each other.

\subsection{Equivalence in a nutshell}

Let us give the key intuition for our equivalence.
First, to convert a prediction market into a CFMM, suppose we have cost function $C$.
We will turn it into a rule for pricing any set of $n$ assets, without money.
Observe that a unit of cash is equivalent to the ``grand bundle'', represented as the all-ones vector $\ones\in\reals^n$, containing one share of each security: exactly one of them will pay off \$1, and the rest \$0, so the grand bundle is worth exactly \$1.
We can therefore simulate transactions in a prediction market with only the $n$ assets (securities) and no cash: a trade of $\r\in\reals^n$ in exchange for $c\in\reals$ units of cash can be expressed as the combined trade $\r' = \r - c\ones$.

We can now simply set $\varphi(\q) = -C(-\q)$ and, perhaps surprisingly, obtain a CFMM where valid trades are equivalent to those allowed by the original cost function.
The negations resolve a difference in sign conventions (i.e., whether $\q$ is interpreted as a gain or loss) and convert the convex $C$ to a concave $\varphi$.
We will show that concavity is equivalent to the desirable trading axiom of \emph{demand responsiveness}; and $\varphi$ also satisfies other nice trading axioms, like \emph{liquidation}: the trader can purchase any bundle of assets for some amount of any given asset.
Perhaps surprisingly, these and other market-making axioms come ``for free'', as long as $C$ is a good prediction-market cost function.
One important exception is that ideally CFMMs do not deplete their reserves.
As described above, we will fix this issue with a more sophisticated construction based on the perspective transform that guarantees nonnegative reserves.
Figure~\ref{fig:intro-figure-lmsr} illustrates that construction applied to LMSR.

In the other direction, one can turn any CFMM into a prediction market with good elicitation properties.
Given a CFMM with potential function $\varphi$, we show that it has an equivalent cost function defined as follows: Let $C(\q)=c$ for the constant $c$ satisfying $\varphi(c\ones - \q) = \varphi(\qZ)$ where $\qZ\in\reals^n$ is the vector of initial reserves and again $\ones\in\reals^n$ is grand bundle (the all-ones vector).
For example, as we recover in \S~\ref{sec:equivalence-examples}, the following cost function is equivalent to Uniswap $\varphi(\q)= \sqrt{q_1 q_2}$ when the initial reserves satisfy $\varphi(\qZ) = k$,
\begin{equation}
  \label{eq:uniswap-cost-intro}
    C(\q) = \frac{1}{2}\left(q_1+q_2+\sqrt{4k^2 + (q_1-q_2)^2}\right)~.\\
\end{equation}
In fact, by standard prediction market facts, this general result implies that every ``good'' CFMM can be converted into a \emph{proper scoring rule} \citep{gneiting2007strictly};
we derive the scoring rule for Uniswap in \S~\ref{sec:equivalence-examples}.

\subsection{Related work}

The literature on CFMMs contains several axiomatic results in a spirit related to our work in \S \ref{sec:characterization}, e.g. \citet{angeris2020improved,bichuch2022axioms,schlegel2022axioms}.
For example, \citet{bichuch2022axioms} characterize the set of ``ideal'' CFMMs from among the class of all CFMMs and \citet{schlegel2022axioms} axiomatically characterize some sub-classes of CFMMs.
A major difference between our axiomatic results in \S~\ref{sec:characterization} and the above works is that we do not take the structure of CFMMs as given, for example, existence of a potential/utility function or its concavity.
We derive these attributes from the axioms directly.

Several authors have noted pieces of the equivalence we present.
Most technically related is \citet{chen2007utility}, which proposes the constant-utility market maker for prediction markets.
Their market maker accepts trades that, according to some subjective fixed belief (probability distribution) and risk-averse utility function, maintain its expected utility at a constant value.
Intriguingly, they observe in their eq.\ (14) that such a market maker with log utility is equivalent to a cost function of the form~\eqref{eq:uniswap-cost-intro}, which turns out to be the result of our conversion of Uniswap to a prediction market; see \S~\ref{sec:equivalence-examples}.
We emphasize that \citet{chen2007utility}, unlike the current paper, focused only on prediction markets, i.e. do not show how to use their market maker for more general classes of assets.
Within the prediction market context, they also only give equivalences to a certain class of cost functions, those corresponding to weighted pseudo-spherical scoring rules.

Other works have specifically connected market makers from decentralized finance and prediction markets.
For example, \citet{paradigm2021universal} discusses how to apply the functional format of the LMSR as a potential function in a CFMM.
Similarly, \citet{manifold2022maniswap} and \citet{gnosis2022automated} apply Uniswap to the case where assets are contingent securities, obtaining a prediction market.
Such works can be described as applying a functional form from one context to a different context, without justifying why they will perform well in the new context.
In contrast, this work shows how to convert any cost function to a CFMM and vice versa via a reduction that guarantees to preserve axiomatic properties of the original.

Other technical tools come from the prediction market literature.
We rely specifically the cost-function market formulation \citep{chen2007utility,abernethy2013efficient}, axiomatic approaches to prediction markets \citep{abernethy2013efficient,abernethy2014general,frongillo2018axiomatic}, and the duality between prediction markets and scoring rules~\citep{hanson2003combinatorial,abernethy2014general}.
We also make use of constructions from the literature on financial risk measures, specifically obtaining a convex risk measure from a set of ``acceptable'' positions~\citep{follmer2008convex,follmer2015axiomatic}; the connection between risk measures and prediction markets has been noted several times~\citep{chen2010new,agrawal2011unified,frongillo2015convergence}.

\subsection{Notation}
Vectors are denoted in bold, e.g. $\q \in \reals^n$.
The all-zeros vector is $\0 = (0,\dots,0)$ and the all-ones vector is $\ones = (1,\dots,1)$.
The vector $\bm{\delta}_i = (0,\dots,1,\dots,0)$, i.e. all-zero except for a one in the $i^{th}$ position.

Comparison between vectors is pointwise, e.g. $\q \succ \q'$ if $q_i > q'_i$ for all $i=1,\dots,n$, and similarly for $\succeq$. 
We say $\q \succneqq \q'$ when $q_i \geq q'_i$ for all $i$ and $\q \neq \q'$.
Define $\reals^n_{\geq 0} = \{\q \in \reals^n \mid \q \succeq \0\}$, $\reals^n_{>0} = \{\q \in \reals^n \mid \q \succ \0\}$, et cetera.

Let $f: \reals^n \to \reals$.
We will use the following conditions.
\begin{itemize}
\item \emph{increasing}:\; $f(\q) > f(\q')$ for all $\q,\q'\in\reals^n$ with $\q \succneqq \q'$.
\item \emph{convex}:\; $\forall x,y\in \reals^n, \lambda\in[0,1]$, $f(\lambda\cdot x+(1-\lambda)y)\leq \lambda f(x)+(1-\lambda)f(y)$.
\item \emph{concave}:\; $-f$ is convex.
\item \emph{quasiconcave}:\; $\forall c \in \reals$, the set $\{x \mid f(x) \geq c\}$ is convex.
\item \emph{$\ones$-invariant}:\; $f(\q+\alpha\ones) = f(\q) + \alpha$ for all $\q\in\reals^n$, $\alpha\in\reals$.
\item \emph{1-homogeneous} (on $\reals^n_{>0}$):\; $f(\alpha\q) = \alpha f(\q)$ for all $\q\succ\0$ and $\alpha>0$.
\end{itemize}
Given a list of vectors, e.g. $h = (\vec{r_1},\dots,\vec{r_t})$, the sum is denoted $\hsum(h) = \vec{r_1} + \cdots + \vec{r_t}$. 
Concatenation of a new trade $\r$ onto a history $h$ is denoted $h \oplus \r$.

\section{A Characterization of Fixed-Liquidity CFMMs}\label{sec:characterization}

To lay a foundation, we consider market makers that have ``fixed liquidity'': they do not accept loans nor charge transaction fees.
The fixed-liquidity case allows for a simple, universal characterization: if one wants an automated market to satisfy certain natural axioms, it \emph{must} be a CFMM as we show in Theorem \ref{thm:itsacfmm}.

\subsection{Automated market makers for general asset markets}

We consider automated market makers that process trades sequentially.
There are $n$ assets $\{1,\dots,n\}$.
A \emph{trade} or \emph{bundle} $\r \in \reals^n$ represents $r_i$ net units of each asset $i$ being given to the market maker by a trader.
A positive $r_i$ represents a net transfer from the trader to the market maker, and negative $r_i$ represents a net transfer from the market maker to the trader.
A \emph{history} $h$ is an ordered list of trades.
The empty history is denoted $\epsilon = ()$.

A market maker is defined by a function $\valtrades(h)$ that specifies, for any valid history $h$, the set of valid trades that it is willing to accept from the next participant. 
For each arriving participant $t=1,\dots$:
\begin{itemize}
  \item The current market history is $h_{t-1}$.
  \item The valid trades are given by $\valtrades(h_{t-1}) \subseteq \reals^n$.
  \item Participant $t$ selects some trade $\r_t \in \valtrades(h_{t-1})$.
  \item The market history is now $h_t = h_{t-1} \oplus \r_t$.
\end{itemize}
The market maker begins with some \emph{initial reserves} $\qZ$.
It will be convenient to include the initial reserves in the history, so in the above we would have $h_0 = (\qZ)$.
At any history $h$, we define the notation
  \[ \qh = \hsum(h) ~, \]
meaning the current reserves are equal to the sum of the trades in the history.

Given a market maker defined by $\valtrades$ and a history $h$, we let $\valhist(h)$ denote the \emph{valid histories extending $h$}, i.e.\ all valid histories with $h$ as a prefix.
Formally, $\valhist(h)$ is the smallest set containing $h$ such that, for all $h' \in \valhist(h)$ and all $\r \in \valtrades(h')$, we have $h' \oplus \r \in \valhist(h)$.
In particular, $\valtrades(\epsilon)$ may be interpreted as the set of valid initial reserves for this market maker, and $\valhist(\epsilon)$ the set of valid histories.
However, we abuse notation slightly in that $\valhist(\epsilon)$ is defined not to include $\epsilon$ itself, i.e. it only includes nonempty histories.

\paragraph{Examples and CFMMs.}
The simplest example of an automated market maker is when $\valtrades(h)$ is a constant set that does not depend on $h$.
Such a market maker would offer a fixed exchange rate between assets.
However, we would like the exchange rates to adapt depending on demand.
In the context of decentralized finance, a common approach is the following.
\begin{definition}[CFMM]
  A \emph{constant-function market maker (CFMM)} is a market such that
  \[ \valtrades(h) = \left\{ \r \in \reals^n \mid \varphi(\r + \qh) = \varphi(\qh) \right\} , \]
  for some function $\varphi: \reals^n \to \reals$, called the potential function.
\end{definition}
In a CFMM, the market maker accepts any trade keeping the potential of its reserves $\varphi(\qh)$ constant.
Some common examples are Uniswap or the constant product market where $\varphi(\r)=(r_1\cdot r_2\cdots r_n)^{1/n}$, and its generalization Balancer or constant geometric mean market which has $\varphi(\r)=r_1^{\alpha_1}\cdot r_2^{\alpha_2}\dots r_n^{\alpha_n}$ with $\sum_{i=1}^n \alpha_i =1, \alpha_i \geq 0 \forall i$.

\subsection{Axioms}\label{subsec:axioms}

An axiom is a potentially desirable property for a market maker as defined by $\valtrades$.
We begin with a few intuitive axioms.

\begin{axiom}[\nodominatedtrades]
  For all $h\in\valhist(\epsilon)$, $\valtrades(h)$ does not contain any \emph{dominated trade}, i.e. an $\r$ where, for some $\r' \in \valtrades(h)$, we have $\r' \succneqq \r$.
\end{axiom}

Assuming all assets have nonnegative value, a rational trader would never select a dominated trade.
Conversely, as typically $\0\in\valtrades(h)$, meaning a trader can chose not trade altogether, then under \nodominatedtrades no trader can take assets from the market maker without giving something in return.
This latter property is the typical condition of \emph{no-arbitrage} (within the system) in the finance literature.

\begin{axiom}[\pathindependence]
  For all $h \in \valhist(\epsilon)$, for all $\r \in \valtrades(h)$ and $\r'\in \valtrades(h \oplus \r)$, we have that $\r+\r'\in \valtrades(h)$.
\end{axiom}
\pathindependence ensures that a trader cannot profit by splitting a single large trade into two smaller sequential trades instead.
By backward induction, it implies that any \emph{sequence} of trades can just as well be made in a single bundle:
if $h' = h \oplus r_1 \oplus \cdots \oplus r_j$ is a valid history, we must have $\r_1 + \cdots + \r_j \in \valtrades(h)$.

\begin{observation} \label{obs:PI-strong}
  \pathindependence implies the following inductive version: for all $h \in \valhist(\epsilon)$ and all $h' \in \valhist(h)$, we have $\q_{h'} - \q_h \in \valtrades(h)$.
\end{observation}

\begin{axiom}[\liquidation]
For all $h\in\valhist(\epsilon)$ and all $\r,\r' \succneqq \0$, there exists $\beta \geq 0$ such that $\r - \beta \cdot \r' \in \valtrades(h)$.
\end{axiom}
In other words, a participant can supply any nonnegative bundle $\r$ and specify a nonnegative demand bundle $\r'$, and receive some multiple of $\r'$ in return for $\r$.
The \liquidation axiom captures the point of a ``market maker'', i.e., to offer to trade at some exhange rate between any pair of assets (or more generally, bundles).

Any market satisfying \liquidation and \nodominatedtrades must in particular offer the option to trade nothing; to see this set $\r = \r'$ for any nontrivial bundle $\r$.

\begin{observation} \label{obs:0-valtrades}
  \liquidation and \nodominatedtrades imply $\0 \in \valtrades(h)$ for all $h \in \valhist(\epsilon)$.
  Furthermore, any $\r \in \valtrades(h) \setminus\{\0\}$ can be written $\r = \r^+ - \r^-$ where $\r^+ = \max(\r,\0)$ and $\r^- = \min(\r,\0)$ satisfy $\r^+,\r^- \succneqq \0$.
\end{observation}

\begin{axiom}[\demandresponsiveness]
  For all $h \in \valhist(\epsilon)$, if $\r-\r' \in \valtrades(h)$ for some $\r,\r' \succeq \vec{0}$, and if there exist $\alpha,\beta > 0$ such that $\alpha \r - \beta \r' \in \valtrades(h \oplus (\r-\r'))$, then $\beta \leq \alpha$.
\end{axiom}
In other words, if a participant supplies $\r$ in return for $\r'$ once, then the ``exchange rate'' should increase the second time in a row: more of $\r$ is needed to buy a corresponding amount of $\r'$.

\subsection{Characterization of fixed-liquidity CFMMs}

In this subsection we characterize automated market makers that satisfy the axioms above: CFMMs with increasing, concave potential functions.
In addition to setting the stage for the equivalence with prediction markets, this result places limits on the design space of automated market makers, or at least shows what axioms must be relinquished to move beyond the concave-CFMM framework.
We first prove that \liquidation, \nodominatedtrades, and \pathindependence already imply the market can be implemented as a CFMM for \emph{some} potential function $\varphi$.
\demandresponsiveness then gives concavity.
Finally, we extend these results from fixed initial market reserves to show that, when considering multiple initial reserves, the market is still without loss of generality a CFMM for some \emph{quasiconcave} potential.

\begin{lemma} \label{lemma:fixedq-axioms-cfmm}
  Fix the initial reserves $\q_0$.
  If a market maker $\valtrades$ with $\valtrades(\epsilon) = \{\q_0\}$ satisfies \liquidation, \nodominatedtrades, and \pathindependence, then it is a CFMM for some potential function $\varphi$.
\end{lemma}
The key point of the proof is that there is a fixed set of reachable states of the market $\R$, and any of these is always reachable in a single trade regardless of the history.
We can then define $\varphi$ to vanish exactly on that set.
\begin{proof}
  Let $h$ be any valid history with initial reserves $q_0$.
  We first prove that $h \in \valhist((\q_0)) \iff \qh-\q_0 \in \valtrades((\q_0))$, by induction on the length of $h$.
  The base case $h = (\q_0)$ is immediate.
  For the inductive step, let $h \in \valhist(\epsilon)$ and consider a trade $\r$.
  We show $\r \in \valtrades(h) \iff \r + \qh -\q_0 \in \valtrades((\q_0))$.
  
  ($\implies$)
  Let $\r \in \valtrades(h)$.
  By Observation \ref{obs:PI-strong}, $\r + \qh -\q_0 \in \valtrades((\qZ))$.

  ($\impliedby$)
  Let $\r$ be some trade such that $\r+\qh-\q_0 \in \valtrades((\qZ))$.
  By \nodominatedtrades, we must have $\r_+,\r_- \succneqq \0$, as otherwise $\r + \qh-\q_0$ would dominate $\0 + \qh-\q_0$ or vice versa.
  By \liquidation, for some $\beta \geq 0$, $\r_+ - \beta \r_- \in \valtrades(h)$.
  We will show $\beta = 1$, implying $\r \in \valtrades(h)$.
	By Observation \ref{obs:PI-strong}, $\r_+ - \beta \r_- + \qh -\q_0 \in \valtrades((\qZ))$.
  If $\beta \neq 1$, this trade is dominated by or dominates $\r+\qh -\q_0\in \valtrades((\qZ))$.
  Therefore, by \nodominatedtrades, $\beta = 1$.

  Now we show the market is a CFMM.
  Define $\R = \{\r + \qZ \mid \r \in \valtrades((\qZ))\}$, the set of valid reserves after a single trade.
  We proved above that $h \in \valhist(\epsilon) \iff \qh - \q_0 \in \valtrades((\q_0))$, which holds if and only if $\qh \in \R$.
  Therefore, we can choose any potential function with $\R$ as a level set, e.g. $\varphi(\q) = 0$ if $\q \in \R$ and $\varphi(\q) = 1$ otherwise.
\end{proof}

In the proof of Lemma \ref{lemma:fixedq-axioms-cfmm}, $\R$ consists of the level set containing $\q_0$, i.e. the set of valid reserves the market is willing to maintain starting from $\q_0$. 
In general, given a function $\varphi$ and initial reserves $\q_0$, it will be useful to define the following notation:
\begin{align*}
  \R_{\varphi}(\q_0)   &:= \left\{\q \in \reals^n \mid \varphi(\q) = \varphi(\q_0) \right\} ,  \\
  \R^+_{\varphi}(\q_0) &:= \left\{\q + \r \mid \q \in \R_{\varphi}(\q_0), \r \succeq \vec{0} \right\} .
\end{align*}
Here, $\R^+_{\varphi}(\q_0)$ consists intuitively of valid reserves, plus all reserves that are weakly preferred by the market maker to some valid reserve.
The key geometrical point is that, under \demandresponsiveness, $\R^+_{\varphi}(\q_0)$ is a convex set:

\begin{lemma} \label{lemma:fixedq-DR-Rplusconvex}
  Fix the initial reserves $\q_0$.
  If a market maker $\valtrades$ with $\valtrades(\epsilon) = \{\q_0\}$ satisfies \liquidation, \nodominatedtrades, \pathindependence, and \demandresponsiveness, then it is a CFMM for some potential $\varphi$ and, furthermore, $\R^+_{\varphi}(\q_0)$ is a convex set.
\end{lemma}
The proof of Lemma~\ref{lemma:fixedq-DR-Rplusconvex}, as well as the rest of the claims in this section, can be found in Appendix~\ref{sec:cfmm-lemmas}.
We obtain that, for a fixed initial starting point, a market maker satisfying our axioms can always be implemented as a CFMM with a concave potential.
\begin{proposition} \label{prop:fixedq-cfmm}
  Fix the initial reserves $\q_0$.
  A market maker $\valtrades$ with $\valtrades(\epsilon) = \{\q_0\}$ satisfies \liquidation, \nodominatedtrades, \pathindependence, and \demandresponsiveness if and only if it can be implemented as a CFMM with an increasing, concave potential function $\varphi$.
\end{proposition}

If the market maker has multiple initial feasible reserves, then the market initiated from each one can be implemented as a CFMM with a concave potential.
But it is not clear that there is a single concave potential function that implements all of these markets simultaneously.
However, there at least exists a single \emph{quasiconcave} potential function that implements the entire market as a CFMM simultaneously.
This statement holds as long as we assume the slightly stronger axiom of \strongpathindependence, which enforces a relationship between the market when starting at different initial reserves.

\begin{axiom}[\strongpathindependence]
  For all $h,h' \in \valhist(\epsilon)$ with $\hsum(h) = \hsum(h')$, we have $\valtrades(h) = \valtrades(h')$.
  Furthermore, we have $\q_h \in \valtrades(\epsilon)$ for all $h \in \valhist(\epsilon)$.
\end{axiom}

\begin{theorem}\label{thm:itsacfmm}
  A market maker, with $\valtrades(\epsilon)=\reals^n$, satisfies \liquidation, \nodominatedtrades, \strongpathindependence, and \demandresponsiveness if and only if it is a CFMM with an increasing, continuous, quasiconcave potential function $\varphi$.
\end{theorem}

As the potential functions in Proposition~\ref{prop:fixedq-cfmm} are concave rather than continuous and quasiconcave, a weaker condition, one might ask whether this weaker condition is necessary in Theorem~\ref{thm:itsacfmm}.
We conjecture that the answer is yes, in the sense that there are CFMMs with continuous, quasiconcave potential functions which cannot be ``concavified'' by any monotone transformation; see e.g.~\citet{fenchel1953convex,connell2017concavifying}.

\section{Equivalence of CFMMs and Prediction Markets}
\label{sec:equivalence-cfmms-pred}

The goal of CFMMs is to facilitate trade.
In contrast, a prediction market is designed to to elicit and aggregate information about future events.
We will now see that prediction markets (satisfying good elicitation axioms) are in a very strong sense equivalent to CFMMs (satisfying good trade axioms).
We give direct reductions to convert each into the other.
As a consequence, much existing research on design and properties of prediction markets can now transfer to CFMMs.

In \S~\ref{sec:define-cost}, we define cost function prediction markets and recall that they characterize automated market makers satisfying elicitation axioms.
In \S~\ref{sec:cashless-pred-market}, we take the key step of converting them into ``cashless prediction markets''.
In \S~\ref{sec:equivalence}, we give the reductions and equivalence between prediction markets and CFMMs.
\S~\ref{sec:cfmms-scoring-rules} discusses information-elicitation consequences, and \S~\ref{sec:equivalence-examples} gives examples.

\subsection{Cost function prediction markets}
\label{sec:define-cost}

To begin, let us briefly review prediction markets for Arrow--Debreu (AD) securities, which are designed to elicit a full probability distribution over a future event.
Formally, let $\Y = \{y_1,\ldots,y_n\}$ be a set of outcomes, and let $Y$ be the future event, a random variable taking values in $\Y$.
For example, in a tournament with $n$ teams, $Y$ can be the winner of the tournament, and $Y=y_i$ is the outcome where team $i$ wins.
The probability simplex on $\Y$ is denoted $\Delta_{\Y}$.
The AD prediction market is a market with $n$ assets $A_1,\dots,A_n$.
When the outcome of $Y$ is observed, each unit of $A_i$ pays off to the owner one unit of cash (some fixed currency, such as US dollars) if $Y=y_i$, and pays off zero otherwise.\footnote{In more general prediction markets than the Arrow--Debreu case described here, there can be any number of assets that each pay off according to an arbitrary function of $Y$.}
We allow traders to both buy and sell assets, including holding negative amounts, representing a short position.

By design, a trader's fair price for $A_i$ is therefore the probability $\Pr[Y = y_i]$ according to their belief.
The current market prices for $A_1,\dots,A_n$ form a ``consensus'' prediction in the form of a probability distribution over $\Y$.

A priori, as with general asset markets, it may seem that the valid trades and associated prices at each moment can be set in essentially any manner whatsoever.
It may therefore be surprising that, in order to successfully elicit and aggregate information while maintaining path independence, a market for contingent securities must take the form of a \emph{cost-function market maker}, as we define next.
We then sketch the proof of this well-known characterization in Fact~\ref{fact:elicitation-cost-function}.

The classic cost-function market maker for Arrow--Debreu securities, slightly generalized for our setting, operates as follows.
\begin{definition} \label{def:cost-function}
  Let $C: \reals^n \to \reals$ be convex, increasing, and $\ones$-invariant.
  Let $\qZ \in \reals^n$ be an initial state.
  The \emph{AD market maker defined by $C$}, with initial state $\qZ$, operates as follows.
  At round $t=1,2,\dots$
  \begin{itemize}
  \item A trader can request any bundle of securities $\r_t \in \reals^n$.
  \item The market state updates to $\q_t = \q_{t-1} + \r_t$.
  \item The trader pays the market maker $C(\q_t) - C(\q_{t-1})$ in cash.
  \end{itemize}
  After an outcome of the form $Y=y_i$ occurs, for each round $t$, the trader responsible for the trade $\r_t$ is paid $(\r_t)_i$ in cash, i.e. the number of shares purchased in outcome $y_i$.
\end{definition}
When $C$ is differentiable, the instantaneous prices of the securities are given by $\p_t = \nabla C(\q_t)$.
That is, the approximate cost of any bundle $x \in \reals^{\Y}$ is given by $\p_t \cdot x$.
(The exact cost is given by integrating this trade from zero to $x$, resulting in $C(\q_t + x) - C(\q_t)$.)
We can therefore regard $\p_t = \nabla C(\q_t)$ as the market prediction.
Because of $\ones$-invariance, $\p_t$ is always a probability distribution.

We can cast cost-function market makers as a special case of the automated market maker framework in \S~\ref{sec:characterization}.
Specifically, represent cash by an asset $A_\$$, and consider the $n+1$ assets $\A = \{A_1,\ldots,A_n,A_\$\}$. 
Then we can restate Definition~\ref{def:cost-function} as the following asset market $\valtrades_C$ over $\A$.
An important convention is that \textbf{prediction markets interpret states $\q$ and trades $\r$ as net transfers to the trader}, whereas our automated market maker follows the CFMM convention to interpret $\q$ and $\r$ as net amounts for the market maker.
Therefore, a negative sign is always required to translate between the settings.
\begin{itemize}
  \item The initial reserves are $(-\qZ,0) \in \reals^{n+1}$.
  \item Define $\valtrades_C(h) = \left\{ \left(\r,\, \alpha\right) \mid \r \in \reals^n, \alpha = C(-\qh - \r) - C(-\qh) \right\}$.
\end{itemize}
In other words, $\valtrades_C(h)$ consists of any bundle of securities $\r \in \reals^n$ to be sold, along with a payment in cash of $C(-\qh - \r) - C(-\qh)$.
As mentioned, the input to $C$ is negated because $C$ expects as an argument the net bundle the market has \emph{sold} rather than purchased.

A priori, the design space for prediction markets could include any automated market maker over any set of contingent securities.
From this perspective, it is perhaps surprising that, to achieve ``good'' elicitation of predictions according to a standard set of axioms, one \emph{must} use a cost-function market maker.
That is, the following is known:

\begin{fact}[\cite{frongillo2018axiomatic,abernethy2013efficient}]
  \label{fact:elicitation-cost-function}
  Let $M$ be any automated market maker offering contingent securities for cash and satisfying these axioms:
  (1) Path Independence;
  (2) Incentive Compatibility for predicting $Y$, in the sense that:
    (2a) Any history $h$ of the market defines a probability distribution over $Y$ (the ``market prediction''), and
    (2b) All possible predictions are achievable by some history, and
    (2c) A trader maximizes expected net payment by moving the market prediction match their belief $p \in \Delta_Y$.
  Then $M$ is a cost-function AD prediction market.
\end{fact}
Fact \ref{fact:elicitation-cost-function} is essentially a characterization: any cost-function AD prediction market satisfies the above axioms, if $C$ is also differentiable with a certain set of gradients~\cite{abernethy2013efficient}.
For completeness, Appendix \ref{app:elicitation} includes a sketch of a proof.
The key ideas, synthesized from \cite{frongillo2018axiomatic,abernethy2013efficient}, are these: Incentive Compatibility requires the market to use \emph{proper scoring rules}, functions that induce truthful forecasts from individual experts (see Section \ref{sec:cfmms-scoring-rules}).
Path Independence then implies that market must use ``chained'' scoring rules as proposed by \cite{hanson2003combinatorial}.
Finally, it is known that Hanson's scoring-rule markets are equivalent to cost function based prediction markets via convex duality~\cite{abernethy2013efficient}.

\subsection{Cashless prediction markets}
\label{sec:cashless-pred-market}

A key step to translating between prediction markets and CFMMs is this observation: without loss of generality, an AD prediction market can drop cash from its assets $\A = (A_1,\dots,A_n,A_{\$})$.
Because exactly one of the outcomes of $\Y = \{y_1,\dots,y_n\}$ will occur, the ``grand bundle'' of securities $\r = \ones = (1,\dots,1) \in \reals^n$ is worth exactly one unit of cash to all traders; whichever outcome $Y=y_i$ occurs, the $i$th security will pay out one unit of cash and the rest will pay out zero.

Specifically, we can define a ``cashless prediction market'' for any cost function $C$ by taking a demand bundle $\r \in \reals^n$ which would have had a price $\alpha$, and subtracting $\alpha$ units of each security, becoming the bundle $(\r_1 - \alpha, \dots, \r_n - \alpha) = \r - \alpha \ones$.
By $\ones$-invariance, the cost of this new trade must be zero, eliminating the need for cash.
With the negative sign to convert to the automated market maker framework, define for any cost function $C$:
\begin{align}
  \valtrades'_{C}(h) = \left\{ \r + \alpha \ones \mid \r \in \reals^n \right\} \text{ where $\alpha = C(-\qh - \r) - C(-\qh)$} . \label{eqn:valtradesprime}
\end{align}
In Appendix \ref{app:cashless}, we show the equivalence of $\valtrades_C$ and $\valtrades'_C$ in more formality.
The key step is the following lemma, which will be useful for converting to a CFMM:
\begin{lemma} \label{lemma:tradesprime-cfmm}
  $\valtrades'_C(h) = \{ \r \mid C(-\qh - \r) = C(-\qh) \}$.
\end{lemma}
\begin{proof}
  Let $\r \in \valtrades'_C(h)$, and write $\r = \r' + \alpha \ones$ with $\alpha = C(-\qh - \r') - C(-\qh)$.
  Using ones-invariance, $C(-\qh - \r) = C(-\qh - \r' - \alpha \ones) = C(-\qh - \r') - \alpha = C(-\qh)$.
  Conversely, let $\r$ be such that $C(\qh + \r) = C(\qh)$.
  Then, directly, $\r \in \valtrades'_C(h)$, where we observe $\alpha = 0$.
\end{proof}

\subsection{Equivalence between CFMMs and prediction markets}
\label{sec:equivalence}

We now show that, in a strong sense, CFMMs and AD cost-function market makers are equivalent.
Specifically, every ``good'' prediction market mechanism for events on $n$ outcomes defines, in a straightforward manner, a  ``good'' CFMM on $n$ assets; and vice versa.
Here, a ``good'' prediction market is one that satisfies the elicitation axioms we introduced in Subsection \ref{subsec:axioms}, i.e. a cost-function market.
A ``good'' CFMM is one that satisfies the trading axioms, i.e. is defined by a concave, increasing $\varphi$.
Recall that $\valtrades_{\varphi}$ is the market maker defined by the CFMM for $\varphi$, i.e. $\valtrades_{\varphi}(h) = \{ \r \mid \varphi(\qh + \r) = \varphi(\qh)\}$.
Recall that $\valtrades'_C$ was defined in Equation \ref{eqn:valtradesprime}.
\begin{theorem} \label{thm:cost-to-CFMM}
  Let $C$ define an AD cost-function prediction market.
  Then the function $\varphi: \reals^n \to \reals$ defined by
    \[ \varphi(\q) = -C(-\q)  \]
  is concave and increasing, i.e. defines a CFMM.
  Furthermore, $\valtrades_{\varphi} = \valtrades'_{C}$.
\end{theorem}
\begin{proof}
  As $\varphi$ is the negative of a convex function, it is concave.
  Because $C$ is increasing, $\varphi$ is increasing.
  By Lemma \ref{lemma:tradesprime-cfmm}, $\r \in \valtrades'_C(h) \iff C(-\qh - \r) = C(-\qh) \iff \varphi(\qh + \r) = \varphi(\qh) \iff \r \in \valtrades_{\varphi}(h)$.
\end{proof}

\begin{theorem} \label{thm:CFMM-to-cost}
  Let $\varphi: \reals^n \to \reals$ be concave and increasing, defining a CFMM.
  For any $\qZ \in\reals^n$, the function $C: \reals^n \to \reals$ defined by
  \begin{equation}
    \label{eq:cfmm-to-cost}
    C(\q) = \inf \left\{ c \in \reals \mid \varphi(c\ones - \q) \geq \varphi(\qZ) \right\}
  \end{equation}
  is convex, increasing, and ones-invariant, i.e., defines an AD cost-function prediction market.
  Furthermore, for all $h \in \valhist_{\varphi}((\qZ))$, we have $\valtrades'_C(h) = \valtrades_{\varphi}(h)$.
\end{theorem}
\begin{proof}

  The key observations come from the financial risk-measure literature, specifically ``constant-risk'' market makers~\cite[\S 2.1]{frongillo2015convergence}~\cite{chen2007utility} and standard constructions extending convex functions to be ones-invariant~\cite{follmer2015axiomatic}.
  Specifically, $C$ is increasing because $\varphi$ is increasing, and ones-invariant immediately by construction.
  We observe that the infimum is always achieved because $\varphi$, a concave function defined on all of $\reals^n$, is continuous.
  For convexity, given $\q,\q'$ and $\lambda \in [0,1]$, let $c = C(\q), c' = C(\q')$.
  Let $\q^* = \lambda \q + (1-\lambda)\q'$ and $c^* = \lambda c + (1-\lambda)c'$.
  Then by concavity, $\varphi(c^* \ones - \q^*) \geq \lambda \varphi(c\ones - \q) + (1-\lambda)\varphi(c' \ones - \q') = \varphi(\qZ)$.
  Because $\varphi(c^* \ones - \q^*) \geq \varphi(\qZ)$, we have $C(\q^*) \leq c^*$, as desired.

  Now, $\r \in \valtrades'_C(h)$ if and only if $C(-\qh - \r) = C(-\qh)$ by Lemma \ref{lemma:tradesprime-cfmm}.
  Because $C(-\qZ) = 0$, inductively we obtain $\r \in \valtrades'_C(h) \iff C(-\qh - \r) = 0$.
  By construction of $C$ and because $\varphi$ is increasing, $C(-\qh - \r) = 0 \iff \varphi(\qh + \r) = \varphi(\qZ)$.
  And by definition of a CFMM, $\varphi(\qh + \r) = \varphi(\qZ) \iff \r \in \valtrades_{\varphi}(h)$.
\end{proof}

\subsection{CFMMs and scoring rules}
\label{sec:cfmms-scoring-rules}

Recall that a scoring rule is a function $S(\p,y_i)$ that evaluates the prediction $\p \in \Delta_{\Y}$ given an observed outcome $y_i$.
A scoring rule is \emph{proper} if for all probability distributions $\p$, the expected score when $y_i \sim \p$ is maximized by report $\p$.
It is \emph{strictly} proper if $\p$ is the unique maximizer.
It is well-known (see e.g.\ \citet{gneiting2007strictly}) that a scoring rule is proper if and only if it can be written as an affine approximation to a particular convex \emph{generating function} $G(\p)$.

One implication of the above equivalence is that CFMMs are closely linked to proper scoring rules.
This connection emphasizes the surprising nature of the equivalence between cost functions and CFMMs, namely that CFMMs, while designed for a very different context than probabilistic forecasting, \emph{characterize} the properties needed for eliciting forecasts.
\begin{corollary} \label{cor:CFMM-scoring}
  Every CFMM for a concave, increasing $\varphi$ and initial reserves $\qZ$ is associated with a proper scoring rule.
\end{corollary}
The corollary follows immediately from known prediction-market equivalences between cost functions and scoring rules, discussed in Appendix \ref{app:elicit-fact}.
Briefly, every cost function $C$ corresponds to the scoring rule generated by $G = C^*$, the convex conjugate of $C$, and vice versa.

As alluded to above, for strict properness, one needs additional conditions on $C$, namely that (1) $C$ is differentiable, so that there is a unique market price at any given state, and (2) the set of gradients is equal to the probability simplex, so that any given market price/prediction is achievable by some set of trades.
These carry over into corresponding conditions on $\varphi$, as we explore in Appendix \ref{app:cfmms-elicit-ratios}.

\subsection{Examples}
\label{sec:equivalence-examples}

We now present two examples to illustrate our equivalence thus far.
We discuss these examples again in \S~\ref{sec:bounded-res-examples} along with others.

\paragraph{LMSR $\to$ CFMM}

As a warm-up, let us see how the most popular cost function, the LMSR $C(\q) = b \log \sum_{i=1}^n e^{q_i/b}$, can be interpreted as a CFMM.
Letting $\varphi(\q) = C(-\q)$, we have $\r \in \valtrades(h) \iff \varphi(\q + \r) = \varphi(\q_0)$, where $\q_0$ is the initial reserves.
If $k = \varphi(\q_0)$, then $\r$ is a valid trade if and only if $b \log \sum_{i=1}^n e^{-(q_i+r_i)/b} = k \iff \sum_{i=1}^n e^{-(q_i+r_i)/b} = e^{k/b}$.
For two assets, this equation reduces to $e^{-(q_1+r_1)/b} + e^{-(q_1+r_1)/b} = e^{k/b}$ (cf.\ \citep{paradigm2021universal}).
See \S~\ref{sec:bounded-res-examples} for a version which scales the liquidity $b$ depending on $\q_0$.

\paragraph{Uniswap $\to$ cost function / scoring rule}

One of the most iconic CFMMs is Uniswap, with the potential function $\varphi_U:\reals^2 \to\reals$ given by $\varphi_U(\q) = \sqrt{q_1q_2}$.
Interestingly, $\varphi_U$ is not defined on all of $\reals^n$, and is only increasing on $\reals^n_{>0}$; as such, typically one restricts to the latter space.
As we will see in \S~\ref{sec:reserves-and-liqudity}, this restriction to $\reals^n_{>0}$ is typical for CFMMs used in practice, and does not pose a barrier for our equivalence.
Briefly, as long as we require $\q_0 \succ \0$, the market reserves will stay within $\reals^n_{>0}$, and the characterization from Theorem~\ref{thm:itsacfmm} and equivalence from Theorem~\ref{thm:CFMM-to-cost} will still apply.

Let $\q_0 \succ \0$ and set $k = \varphi(\q_0) > 0$.
Applying Theorem~\ref{thm:CFMM-to-cost}, we obtain the cost function from eq.~\eqref{eq:uniswap-cost-intro},
\begin{align*}
  C_k(\q) = \frac 1 2 \left( q_1 + q_2 + \sqrt{4k^2 + (q_1-q_2)^2} \right)~.
\end{align*}
Let us verify the construction.
One can easily check that $C_k$ is $\ones$-invariant and increasing.
Consequently, it suffices to verify a single level set of $C_k$; we will show $\{\q \mid C(-\q) = 0\} = \{\q \mid \q\succ\0, \varphi_U(\q) = k\}$.
\begin{align*}
  C_k(-\q) = 0
  &\iff \sqrt{4k^2 + (q_1-q_2)^2} = (q_1 + q_2)
  \\
  &\iff 4k^2 + (q_1-q_2)^2 = (q_1+q_2)^2 \text{ and } q_1,q_2 > 0
  \\
  &\iff 4k^2 = 4q_1q_2 \text{ and } q_1,q_2 > 0
  \\
  &\iff k = \sqrt{q_1 q_2} \text{ and } q_1,q_2 > 0~.
\end{align*}
We conclude that, indeed, $\valtrades'_{C_k} = \valtrades_{\varphi_U}$, as promised by Theorem~\ref{thm:CFMM-to-cost}.
In other words, a prediction market run using $C_k$ would behave exactly the same as Uniswap.

As discussed in \S~\ref{sec:introduction}, the expression for $C_k$ appears in \citet[eq.\ (14)]{chen2007utility}.
There they show that a market maker keeping constant $\log$ utility assuming the binary outcome would be drawn from the uniform distribution $\pi = (1/2,1/2)$.
Moreover, as observed by \citet{bichuch2022axioms} and \citet{othman2021new}, constant expected log utility under the uniform distribution gives rise to Uniswap: $(1/2)\log(q_1) + (1/2)\log(q_1) = c \iff \sqrt{q_1 q_2} = e^c$.
(Changing this distribution $\pi$ gives rise to more general forms of Uniswap, such as Balancer, which takes the form $\varphi_\pi(\q) = q_1^{\pi_1} q_2^{\pi_2}$.)
Amazingly, it appears that no one has yet made these two connections together, that Uniswap is equivalent to a cost-function prediction market.

In light of Corollary~\ref{cor:CFMM-scoring}, each level set of Uniswap must therefore be equivalent to a scoring rule market for some choice of proper scoring rule.
Let us now derive this family of scoring rules.
Taking the convex conjugate of $C_k$, we obtain the convex generating function
\begin{align*}
  G_k(\p) = -2 \sqrt{ k p_1 p_2 }~.
\end{align*}
The corresponding scoring rule is
\begin{align*}
  S_k(\p,y_i) = G(\p) + dG_\p (\bm{\delta}_i - \p) = - k \sqrt{ p_i / p_j }~,
\end{align*}
where $\{i,j\} = \{1,2\}$.
This scoring rule is exactly the boosting loss of \citet{buja2005loss}, negated and scaled by $k$.
(This scoring rule also appears in a shifted form as ``$\mathrm{hs}$'' in \citet{ben2020new}.)
Interestingly, while in general each level set of $\varphi_U$ could have corresponded to an entirely different scoring rule, instead each level set merely scales the same scoring rule.
In fact, as we will see in \S~\ref{sec:role-1-homogeneity}, this phenomenon is shared by all 1-homogeneous CFMMs.

\section{Bounded reserves and liquidity}\label{sec:reserves-and-liqudity}

CFMMs are defined in terms of the current reserves $\q\in\reals^n$, the vector of asset quantities available to trade.
Naturally, one might like to ensure $\q \succeq \0$ at every state of the market, i.e., that the market maker can only trade assets it has, and not go short on any asset.
We now make this restriction explicit.

\begin{axiom}[\boundedreserves]
  $\reals^n_{>0} \subseteq \valtrades(\epsilon)$ and for all $h \in \valhist(\epsilon)$, $\q_h \succeq \0$.
\end{axiom}

A version of the \boundedreserves axiom appeared in \citet{angeris2022constant}, emphasizing its importance as most common potential functions are 1-homogeneous.
In this section, we first characterize all market makers satisfying \boundedreserves in addition to our other four axioms (Theorem~\ref{thm:cfmm-char-boundedreserves}).
We then explore the relationship between \boundedreserves and \emph{bounded worst-case loss} from the prediction market literature, finding the former to be a stronger condition.
In fact, CFMMs constructed from cost functions in our previous equivalence, Theorem~\ref{thm:cost-to-CFMM}, can essentially never satisfy \boundedreserves.
Instead, we present a new construction leveraging a liquidity parameter from the prediction market literature, to convert any cost function to a CFMM satisfying \boundedreserves.
As we will see, this construction recovers every 1-homogeneous CFMM potential function, including nearly all CFMMs used in practice.
Moreover, leveraging Corollary~\ref{cor:CFMM-scoring}, these 1-homogeneous CFMMs are associated with a unique scoring rule, thus allowing us to convert between proper scoring rules and CFMMs satisfying \boundedreserves.

\subsection{Characterizing bounded reserves} 
\label{sec:char-bound-reserv}

We begin with an observation about CFMMs satisfying \boundedreserves: with \liquidation and \nodominatedtrades, the reserves of each asset must be \emph{strictly} positive at all valid histories.
\begin{lemma}
  \label{lem:strictly-positive-reserves}
  A market maker satisfying \boundedreserves, \liquidation and \nodominatedtrades must have $\q_h \succ 0$ for all $h \in \valhist(\epsilon)$.
\end{lemma}
\begin{proof}
  Assume that for a history $h\in \valhist(\epsilon)$, we have $(\q_h)_i=0$ for some $i$.
  For any quantity $t > 0$ of asset $j$ an agent offers, \liquidation states that the market maker would have to quote a price for it in terms of asset $i$.
  As the market maker does not have any of asset $i$, it must quote a price of zero to satisfy \boundedreserves.
  If the agent asks to trade quantity $\frac{t}{2}$ of asset $j$, the market maker can still only give zero units of asset $i$.
  This trade of $\frac t 2$ for zero units of $i$ therefore dominates $t$ for zero, violating \nodominatedtrades.
  Hence the reserves for an asset can never be empty.
\end{proof}

\newcommand{\bdposorth}{\partial \reals^n_{\geq 0}}
We will now characterize CFMMs satisfying all of our axioms.
In light of Lemma~\ref{lem:strictly-positive-reserves}, for any CFMM that satisfies these axioms, it must be the case that $\varphi(\qZ) = \varphi(\q)$ and $\qZ \succ \0$ together imply $\q \succ \0$.
Intuitively, it follows that we may restrict attention to the strictly positive orthant $\reals^n_{>0}$.
While $\varphi$ will be increasing on $\reals^n_{>0}$, however, it might not be on its boundary $\bdposorth := \reals^n_{\geq 0} \setminus \reals^n_{>0}$.
In fact, a prominent example is Uniswap, $\varphi(\q) = \sqrt{q_1q_2}$, where $\varphi(\q) = 0$ for $\q \in \bdposorth$.
Yet the behavior of $\varphi$ in a neighborhood of $\bdposorth$ turns out to be the key to whether \liquidation is satisfied.
To state our results, therefore, we restrict $\varphi:\reals^n_{>0}\to\reals$ without loss of generality, and then use the limits to control the behavior of $\varphi$ near $\bdposorth$.

\begin{theorem}\label{thm:cfmm-char-boundedreserves}
  A market maker satisfies \liquidation, \nodominatedtrades, \strongpathindependence, \demandresponsiveness, and \boundedreserves if and only if the following hold: \,(a)\, $\valtrades(\epsilon) = \reals^n_{>0}$, \,(b)\, the market maker is a CFMM with an continuous, increasing, quasiconcave potential function $\varphi:\reals^n_{>0}\to\reals$, and \,(c)\,~$\lim\limits_{\alpha\to 0^+} \varphi(\q + \alpha\ones) = \lim\limits_{\alpha\to 0^+} \varphi(\alpha\ones)$ for all $\q\in\bdposorth$.
\end{theorem}
\begin{proof}
  First consider the forward direction, and assume we are given a market maker satisfying all 5 axioms.
  For (a), we have $\reals^n_{>0} \subseteq \valtrades(\epsilon)$ by \boundedreserves and the reverse inclusion from Lemma~\ref{lem:strictly-positive-reserves}.
  For (b), we make use of the Lemma~\ref{lemma:Rplusconvex-cfmm}, a more general version of the foreward direction for Theorem~\ref{thm:itsacfmm} (see \S~\ref{sec:cfmm-lemmas}).
  Leveraging (a), the lemma states that the first 4 axioms imply a CFMM with an continuous, increasing, quasiconcave potential $\varphi:\reals^n_{>0}\to\reals$.

  Finally, consider (c).
  Denote $\bar\varphi(\q) := \lim_{\alpha\to 0^+} \varphi(\q + \alpha\ones)$ for $\q \in \reals^n_{\geq 0}$; the limit exists as $\varphi$ is increasing.
  Then condition (c) can be succinctly written: $\bar\varphi(\q) = \bar\varphi(\0)$ for all $\q \in \bdposorth$.
  A fact we will use is that for $\q \in \bdposorth$ and $\r \succ \0$, we have $\bar\varphi(\q) < \varphi(\q+\r)$.
  (Take $c = \min_i r_i > 0$; as $\varphi$ is increasing and by definition of $\bar\varphi$, we have $\bar\varphi(\q) < \varphi(\q + c\ones) \leq \varphi(\q + \r)$.)
  As $\varphi$ is increasing, we have $\bar\varphi(\q) \geq \bar\varphi(\0)$.
  Now suppose for a contradiction we had some $\q \in \bdposorth$ with $\bar\varphi(\q) > \bar\varphi(\0)$.
  By definiton of a limit, there exists $\epsilon > 0$ sufficiently small so that $\bar\varphi(\q) > \varphi(\epsilon \ones)$.
  Let $c = \max_i q_i + 1$ so that $c\ones \succ \q$.
  Then by the fact above we have $\varphi(c\ones) > \bar\varphi(\q) > \varphi(\epsilon\ones)$.
  By continuity of $\varphi$, there exists $c' \in (\epsilon,c)$ satisfying $\varphi(c'\ones) = \bar\varphi(\q)$.
  Let $i$ be an asset such that $q_i = 0$.
  Now let us apply \liquidation at $h = (c'\ones)$ to bundles $r = (c-c')\ones$ and $r' = \bm{\delta}_i$, giving some $\beta \geq 0$ such that $r - \beta r' \in \valtrades(h)$.
  Let $h' = (c'\ones) \oplus (r - \beta r')$ be the resulting history.
  Then $\varphi(c'\ones) = \varphi(\q_{h'})$ by (b).
  By Lemma~\ref{lem:strictly-positive-reserves}, we must have $0 \prec \q_{h'} = c'\ones + (r-\beta r') = c\ones - \beta\bm{\delta}_i$, which implies $\beta < c$.
  As $q_i=0$, we therefore have $\q_{h'} - \q \succ \0$, as the $i$th coordinate is $c-\beta > 0$ and all other coordinates are at least 1.
  By the fact above, $\varphi(\q_{h'}) > \bar\varphi(\q) = \varphi(c'\ones)$, a contradiction.  

  For the reverse direction, let $\varphi$ satisfying conditions (a), (b), and (c) be given.
  Let $\bar\varphi$ be defined as above.
  We first show that we may assume without loss of generality that $\bar\varphi$ is continous on $\reals^n_{\geq 0}$.
  Let $f(\alpha) = \bar\varphi(\alpha\ones)$ for all $\alpha \geq 0$.
  Then $f:[0,\infty)$ is monotone increasing and continuous, and has a monotone increasing and continuous inverse.
  Thus, $\varphi' = f^{-1} \circ \varphi$ satisfies conditions (a), (b), and (c) as well.
  Moreover, $\varphi'$ satisfies $\varphi'(\alpha\ones) = \alpha$ for all $\alpha>0$;
  in other words, $\varphi'$ follows the construction of Theorem~\ref{thm:itsacfmm}.
  Now as $\bar\varphi'(\0) = 0 > -\infty$ by the above, and $\bar\varphi'(\q) = \bar\varphi'(\0)$ for all $\q \in \bdposorth$ by (c), \citet[Proposition 4.1]{crouzeix2005continuity} states that $\bar\varphi'$ is continuous on $\bdposorth$.
  As $\bar\varphi' = \varphi'$ on $\reals^n_{>0}$, we have continuity on $\reals^n_{\geq 0}$.
  
  \boundedreserves is trivially satisfied as $\varphi$ is only defined (or equivalently, equal to $-\infty$) outside of $\reals^n_{>0}$, so $\qZ \succ \0$ and $\varphi(\qh) = \varphi(\qZ)$ immediately implies $\qh \succ \0$.
  It follows that the proofs for \demandresponsiveness, \nodominatedtrades, and \strongpathindependence in Theorem~\ref{thm:itsacfmm} immediately extend to the case where the domain is restricted to $\reals^n_{>0}$, since they all begin with given elements of $\valhist(\epsilon)$.
  (The second condition of \strongpathindependence also follows immediately from $\q_h \succ \0$.)

  It remains therefore to show \liquidation.
  We will use the fact that $\bar\varphi(\q) > \bar\varphi(\0)$ for all $\q\succ\0$, which follows by a similar argument to the above: letting $c = \tfrac 1 2 \min_i q_i > 0$, we have $\bar\varphi(\q) > \bar\varphi(c\ones) > \bar\varphi(\0)$.
  Now let $h\in\valhist(\epsilon)$ and $\r,\r'\succneqq \0$.
  Then $\qh \succ \0$ as $\varphi$ is defined only on $\reals^n_{>0}$, giving $\qh+\r \succ \0$ as well.
  Thus $\bar\varphi(\qh+\r) > \bar\varphi(\qh) > \bar\varphi(\0)$.
  Now let $\beta' > 0$ be the unique real value such that $\qh + \r - \beta'\r' \in \partial\reals^n_{\geq 0}$.
  By (c), we have $\bar\varphi(\qh + \r - \beta'\r') = \bar\varphi(\0)$.
  Thus $\bar\varphi(\qh + \r - \beta'\r') < \bar\varphi(\qh) < \bar\varphi(\qh + \r - 0\cdot\r')$.
  By continuity of $\bar\varphi$ on $\reals^n_{\geq 0}$ (see above), we have some $\beta\in(0,\beta')$ such that $\bar\varphi(\qh) = \bar\varphi(\qh + \r - \beta\r')$.
  As $\beta < \beta'$, we have $\qh + \r - \beta\r' \succ \0$.
  Thus $\varphi(\qh) = \varphi(\qh + \r - \beta\r')$ as well, giving $\r - \beta\r' \in \valtrades(h)$, as desired.
\end{proof}

In Appendix~\ref{sec:cfmm-lemmas}, we prove a more general version of some key steps of Theorem~\ref{thm:cfmm-char-boundedreserves} for more general restrictions $S$ on the reserves other than $S = \reals^n_{>0}$.
While from \boundedreserves this is perhaps the most interesting case, the entire theorem could in principle be generalized to other $S$ which are ``dominant domains'' (Definition~\ref{def:dominant-domain}); for example, condition (c) in the theorem would replace $\bdposorth$ with $\partial S$.

\subsection{Connection to worst-case loss and liquidity}
\label{sec:conn-worst-case}

When assets are contingent securities, it is perhaps more natural to ask how much money the market maker could lose in some run of the market and some outcome of nature, rather than tracking the number of securities ``in reserve''.
This interpretation gives the following axiom, which is commonly found in the literature on prediction markets~\cite{chen2007utility,abernethy2013efficient,abernethy2014general,cummings2016possibilities,frongillo2018bounded}.
Here we regard ``money'' as the grand bundle $\ones$ (cf.\ \S~\ref{sec:cashless-pred-market}).

\begin{axiom}[\worstcaseloss]
 For all $\qZ$, there exists $b \in \reals$ such that for all $h \in \valhist(\epsilon)$ with $h_0 = (\qZ)$, we have $\hsum(h) \succeq \qZ - b\,\ones$.
\end{axiom}

The \boundedreserves axiom can be regarded as a special case of \worstcaseloss, where $b = \|\qZ\|_\infty = \max_i (\q_0)_i$, but the converse does not hold.
Indeed, essentially any interesting $\ones$-invariant cost function satisfying \worstcaseloss does not satisfy \boundedreserves.
To see why, let $h = (\ones)$ and $\r\in\valtrades(h)$ be any non-dominated trade, i.e., with $\r\nsucceq\0$.
In other words, start with one of each asset, and take any nontrivial trade.
Let $i$ be an asset with $r_i < 0$, and now consider $h' = (\alpha\ones)$ where $\alpha = |r_i|/2$.
Then by $\ones$-invariance, we still have $\r \in \valtrades(h')$, yet now $\hsum(h'\oplus \r)_i = \alpha + r_i < 0$, meaning the market maker is negative in its reserves of asset $i$.
Moreover, as our previous construction $\varphi(\q) = -C(-\q)$ from Theorem~\ref{thm:cost-to-CFMM}) 
gives a $\ones$-invariant market, we have yet to see how to convert a prediction market to a CFMM that respects the reserves regardless of where the market is initialized.

Fortunately, we can borrow ideas from the prediction market literature to make our construction reserves-aware.
Specifically, we will use the relationship between worst-case loss and \emph{liquidity}, roughly the amount of assets available to buy or sell around a given price, which has been extensively studied~\cite{chen2007utility,abernethy2013efficient,abernethy2014general,cummings2016possibilities,frongillo2018bounded}.
In the cost-function formulation, there is a simple and direct tradeoff: more liquidity leads to a higher worst-case loss.

Furthermore, there is a simple transformation from convex analysis that changes any given cost function's liquidity: the perspective transform.
Given a convex function $C(\q)$, its \emph{perspective transform} by $\alpha > 0$ is the function $C_\alpha(\q) = \alpha C(\q/\alpha)$.%
\footnote{Often the perspective transform is given by $(\q,\alpha) \mapsto \alpha C(\q/\alpha)$, a form we will use in Proposition~\ref{prop:construction-5-axioms}.}
For example, for $C(\q) = \|\q\|_2^2$, the perspective transform is $\alpha C(\q/\alpha) = \tfrac{1}{\alpha} \|\q\|_2^2$.
As $\alpha \to \infty$, the liquidity increases: the gradient of $C$, i.e.\ the price, changes more slowly in response to a change in quantity $\q$.
But the worst-case loss increases as well.
One way to see the latter is via the convex duality relationship to proper scoring rules: the perspective transform on $C$ is equivalent to a simple scaling of its corresponding proper scoring rule by $\alpha$.

\subsection{CFMM construction satisfying bounded reserves}
\label{sec:cfmm-construction-reserves}

Using the ideas above, we now provide a construction taking any cost function to a CFMM that satisfies \boundedreserves.
The magic is in ``wrapping up'' the level sets of $C_\alpha(\q) := \alpha C(\q/\alpha)$ for different $\alpha$ into a single function $\varphi$, so that starting at $\qZ$ implies a liquidity of $\alpha = \varphi(\qZ)$.
Furthermore, we ensure that $\varphi(\q) \to 0$ as $\q \to \0$, so that a market starting with very low initial reserves quickly moves the price so as not to run out.

\begin{figure}[t]
  \includegraphics[width = 1.8in,trim={5 5 20 5},clip]{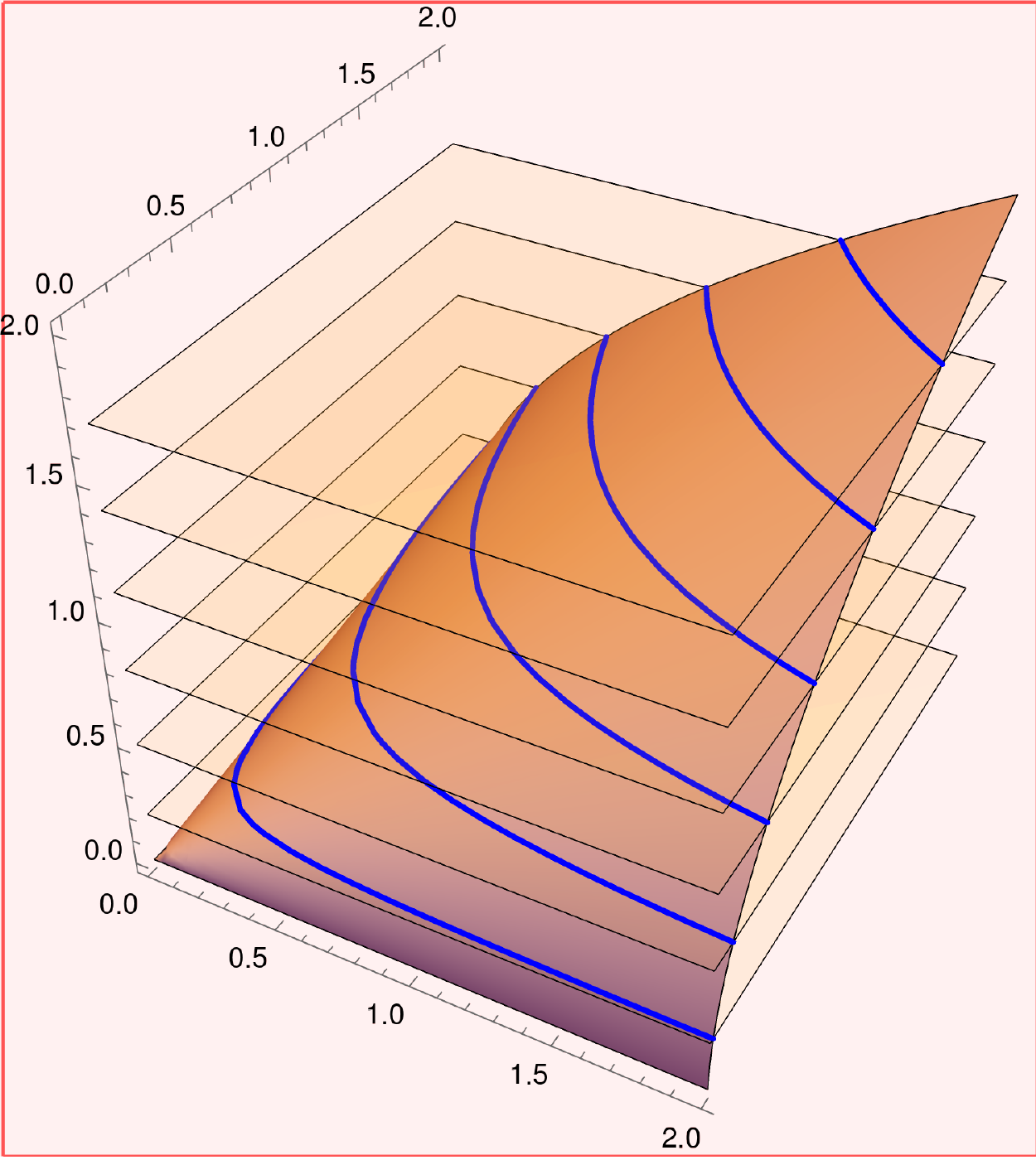}
  \includegraphics[width = 1.8in,trim={5 5 20 5},clip]{figs/perspective-lmsr}
  \includegraphics[width = 1.8in,trim={5 5 20 5},clip]{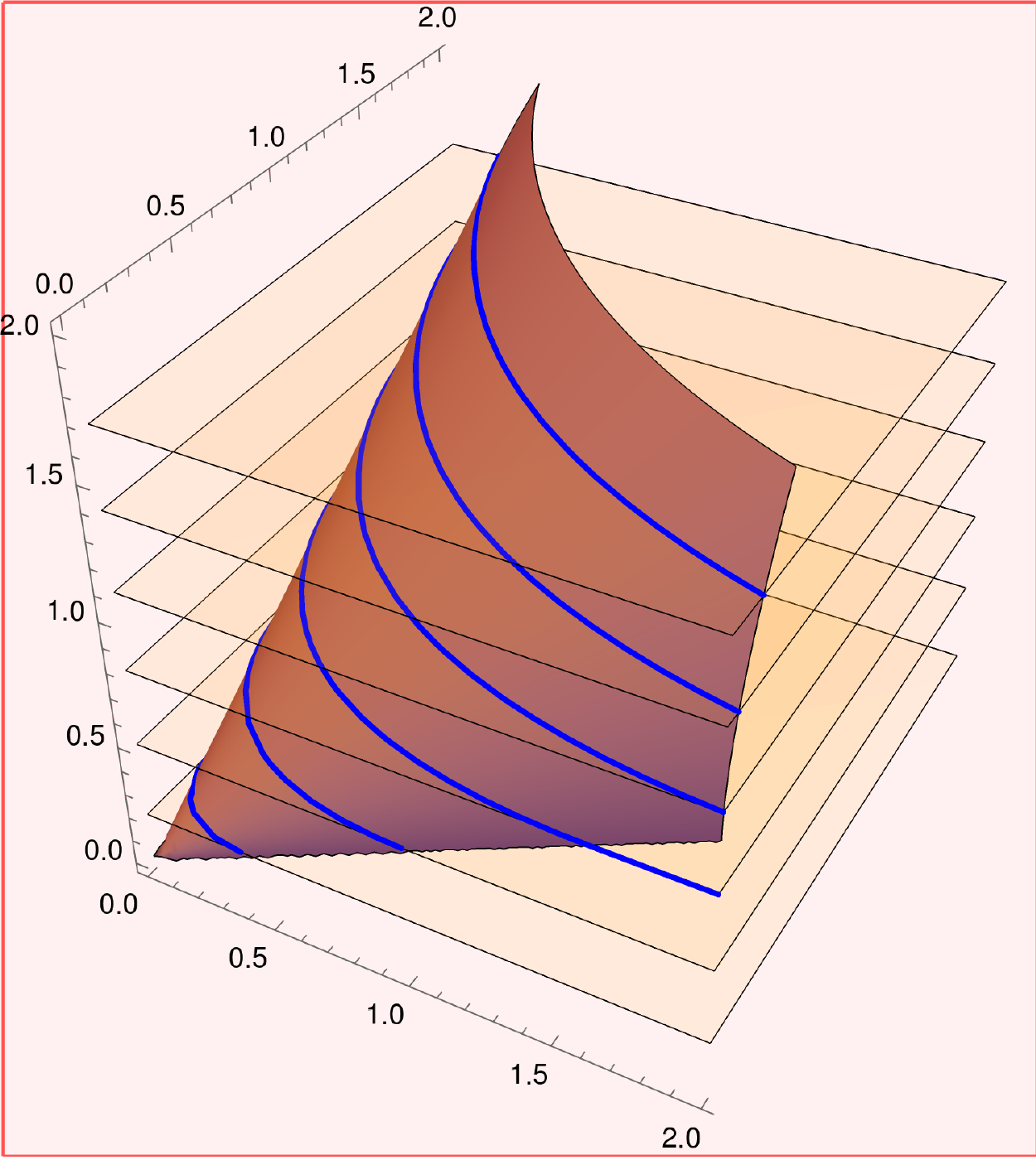}
  \caption{
    Plots of Construction~\ref{cons:bounded-reserves} for Uniswap (L), LMSR (M), and Brier score (R); see \S~\ref{sec:bounded-res-examples} for how the construction was applied in each setting.
    In each plot, the orange surface plots the $0$-level set of the perspective transform $T: (q_1,q_2,\alpha) \mapsto \alpha C(-(q_1,q_2)/\alpha)$, i.e., the set of triples $(q_1,q_2,\alpha)$ such that $\alpha C(-(q_1,q_2)/\alpha) = 0$.
    By design, this set also happens to be the graph of $\varphi$ resulting from Construction~\ref{cons:bounded-reserves}, i.e., the set of triples $(q_1,q_2,\alpha)$ such that $\varphi((q_1,q_2)) = \alpha$.
    To emphasize this connection, the $\alpha$-level sets $\varphi$ for $\alpha = 0.2, 0.6, 1.0, 1.4, 1.8$ are shown in blue.
    (Technically, these level sets should be projected down to the plane on the bottom of the figure.)
    As $T$ is convex (a fact of perspective transforms), the sublevel set $T(q_1,q_2,\alpha) \leq 0$, the region below the surface, is a convex set.
    We can see therefore that $\varphi$ is always a concave function.
  }
  \label{fig:perspective-construction}
\end{figure}

\begin{construction} \label{cons:bounded-reserves}
  Let $C:\reals^n\to\reals$ be a convex, increasing, $\ones$-invariant cost function satisfying $C(\0) > 0$.
  Then define $\varphi: \reals_{>0}^n \to \reals$ by $\varphi(\q) = \alpha$ where $C(-\q/\alpha) = 0$.
\end{construction}

The assumption $C(\0) > 0$ is without loss of generality; otherwise, add a sufficiently large constant to $C$.
Similarly, while this construction centers around the level set $\{\q\mid C(\q) = 0\}$, other level sets may be taken be shifting $C$ by a constant.

Before arguing that this construction is well-defined and satisfies all our axioms, let us verify that it is behaving as we would hope.
Define $C_\alpha(\q) := \alpha C(\q/\alpha)$, which is $C$ but with ``liquidity level'' $\alpha$.
Then by definition the $\alpha$-level set of $\varphi$ matches the $0$-level set of $C_\alpha$, i.e., $\{\q \mid \varphi(\q) = \alpha\} = \{\q \mid C_\alpha(-\q)=0\}$.
(See Figure~\ref{fig:perspective-construction}(L) for a visualization.)
If $\q_0$ satisfies $\varphi(\q_0)=\alpha$ for $\alpha>0$, then we have $\valtrades_\varphi(h) = \valtrades_{C_\alpha}(h)$.
In other words, $\varphi$ indeed has liquidity level $\alpha$ when $\varphi(\q_0) = \alpha$.
Moreover, if $\q_0 = \beta \q$, then $\varphi(\q_0) = \beta \varphi(\q)$, so indeed the liquidity level drops to $0$ as $\q_0 \to \0$, thereby protecting the reserves.
This scaling property is known as \emph{1-homogeneity}, and plays a large role in our analysis; see Lemma~\ref{lem:construction-homogeneous} and \S~\ref{sec:role-1-homogeneity}.

\begin{lemma} \label{lemma:const-welldef}
  The function $\varphi$ from Construction \ref{cons:bounded-reserves} is well-defined.
\end{lemma}
\begin{proof}
  We must show that $\varphi(\q)$ exists and is unique on $\reals^n_{>0}$.
  Let $\q \succ \0$.
  First, define $t = C(\0) > 0$.
  By $\ones$-invariance, $C(- t \ones) = 0$.
  Let $\beta_1 = \min_i q_i$ and $\beta_2 = \max_i q_i$.
  Define $\alpha_1 = \frac{\beta_1}{t}$ and $\alpha_2 = \frac{\beta_2}{t}$.
  We have $-\q/\alpha_1 \preceq -t \ones \preceq -\q/\alpha_2$.
  Because $C$ is increasing, $C(-\q/\alpha_1) \leq C(-t \ones) \leq C(-\q/\alpha_2)$.
  Because $C$ is convex and defined on $\reals^n$, it is continuous.
  Therefore, $C(-\q/\alpha)$ is continuous in $\alpha$ for $\alpha > 0$.
  So there exists $\alpha^* \in [\alpha_1,\alpha_2]$ such that $C(-\q/\alpha^*) = C(-t \ones) = 0$.
  Furthermore, because $C(-\q/\alpha)$ is strictly decreasing in $\alpha$, this solution $\alpha^*$ to $C(-\q/\alpha) = 0$ is unique.
\end{proof}

A useful fact about Construction~\ref{cons:bounded-reserves} is that the resulting potential function $\varphi$ is 1-homogeneous.
Recall that a function $f:\reals^n_{>0}\to\reals$ is \emph{1-homogeneous} if for all $\q\succ\0$ and $\alpha>0$ we have $f(\alpha\q) = \alpha f(\q)$.
Such functions have useful symmetry for market making \citep{othman2010practical}; in particular, since their gradients are 0-homogeneous (invariant under scaling), the instantaneous exchange rates between any two assets remain constant when scaling up the reserves.
In part because of this appealing scale invariance, nearly all popular CFMM potential functions used in practice are 1-homogeneous \citep{angeris2022constant}.

\begin{lemma}
  \label{lem:construction-homogeneous}
  Any $\varphi$ resulting from Construction~\ref{cons:bounded-reserves} is 1-homogeneous, concave, and increasing.
\end{lemma}
\begin{proof}
  Let $C$ be the given cost function.

  (Homogeneous)\;
  Let $\q \succ \0$, and $\alpha = \varphi(\q)$.
  Then $C(-\q/\alpha) = 0$ by definition of $\varphi$ and the fact that it is well defined (Lemma~\ref{lemma:const-welldef}).
  Further, for any $\beta > 0$, we have $C(-\beta\q/(\alpha\beta)) = 0$, which implies $\varphi(\beta\q) = \beta\alpha = \beta\varphi(\q)$.

  (Increasing)\;
  Let $\0 \prec \q' \precneqq \q$.
  Let $\alpha = \varphi(\q)$, so that $C(-\q/\alpha) = 0$.
  As $C$ is increasing, and $-\q'/\alpha \succneqq -\q/\alpha$, we have $C(-\q'/\alpha) > C(-\q/\alpha) = 0$.
  Now let $\alpha' = \varphi(\q')$, so that $C(-\q'/\alpha') = 0$.
  Clearly $\alpha' \neq \alpha$ due to uniqueness in Lemma \ref{lemma:const-welldef}, and if $\alpha' > \alpha$, we would have $-\q'/\alpha \precneqq -\q'/\alpha'$ yet $C(-\q'/\alpha) > 0 = C(-\q'/\alpha')$, contradicting $C$ being increasing.
  We conclude $\alpha > \alpha'$, or equivalently, $\varphi(\q) > \varphi(\q')$, as desired.

  (Concave)\;
  It is well-known that the map $(\q,\alpha) \mapsto \alpha C(-\q/\alpha)$ is convex, as the perspective transform of the convex function $\q \mapsto C(-\q)$.
  Thus, $\{(\q,\alpha) \mid \alpha C(-\q/\alpha) \leq 0\}$ is a convex set as a sublevel set of a convex function.
  Yet this set is precisely the hypograph of $\varphi$, as $\varphi(\q) \geq \alpha \iff C(-\q/\alpha) \leq 0 \iff \alpha C(-\q/\alpha) \leq 0$, so $\{(\q,\alpha) \mid \varphi(\q) \geq \alpha\} = \{(\q,\alpha) \mid \alpha C(-\q/\alpha) \leq 0\}$.
  As the hypograph of $\varphi$ is convex, $\varphi$ is a concave function.
 
\end{proof}

\begin{proposition}
  \label{prop:construction-5-axioms}
  Let $C$ satisfy the conditions of Construction~\ref{cons:bounded-reserves} as well as $\{\q \mid C(\q)=0\} \subseteq \reals^n_{<0}$.
  Then the resulting $\varphi:\reals^n_{>0}\to\reals$ defines a CFMM satisfying \liquidation, \nodominatedtrades, \pathindependence, \demandresponsiveness, and \boundedreserves.
\end{proposition}

\begin{proof}
  We will establish conditions (a), (b), and (c) from Theorem~\ref{thm:cfmm-char-boundedreserves} and then apply the result.
  Part (a) is implied by the domain of $\varphi$.
  Part (b) follows from Lemma~\ref{lem:construction-homogeneous}.
  It thus suffices to show (c).
  
  Let $\bar\varphi:\reals^n_{\geq 0}\to\reals\cup\{-\infty\}$ be the concave closure of $\varphi$.
  Then $\bar\varphi$ is continuous on $\reals^n_{\geq 0}$; see the proof of Theorem~\ref{thm:cfmm-char-boundedreserves}.
  As $\varphi$ is 1-homogeneous from Lemma~\ref{lem:construction-homogeneous}, we have $\bar\varphi(\0) = \lim_{\alpha\to 0^+} \varphi(\alpha\ones) = \lim_{\alpha\to 0^+} \alpha\varphi(\ones) = 0$.
  Now let $\q \succneqq \0$ with $q_i = 0$ for some $i$.
  Suppose for a contradiction that $\bar\varphi(\q) = \alpha > 0$.
  Let $\{\q_k\}_{k\in\mathbb N}$ be a sequence in $\reals^n_{>0}$ converging to $\q$, so that $\bar\varphi(\q_k) \to \alpha$.
  Then letting $\alpha_k = \bar\varphi(\q_k) = \varphi(\q_k)$, we have $-\q_k / \alpha_k \to -\q / \alpha$.
  By definition of $\varphi$, each point $-\q_k / \alpha_k$ is an element of $Z := \{\q'\in\reals^n \mid C(\q')=0\}$.
  As $Z$ is closed (by continuity of $C$), we must have $-\q / \alpha \in Z$.
  Yet by assumption, $Z \subseteq \reals^n_{<0}$, a contradiction as $-\q / \alpha \not\prec \0$ (recall that $q_i = 0$).
  We conclude that $\bar\varphi(\q) \leq 0$.
  As $\varphi > 0$ on $\reals^n_{>0}$, by continuity we also have $\bar\varphi(\q) \geq 0$.
  Thus $\bar\varphi(\q) = 0 \leq 0 = \bar\varphi(\0)$, giving (c).
\end{proof}

\subsection{Homogeneous CFMMs and unique scoring rules}
\label{sec:role-1-homogeneity}

As essentially all CFMMs used in practice are 1-homogeneous \citep{angeris2022constant}, it is natural to ask whether Construction~\ref{cons:bounded-reserves} could have produced them.
We now show the answer is yes: every 1-homogeneous (increasing, concave) potential function is the result of Construction~\ref{cons:bounded-reserves} for some choice of cost function.

\begin{proposition}
  \label{prop:construction-1-homog}
  A potential $\varphi:\reals^n_{>0}$ is the result of applying Construction~\ref{cons:bounded-reserves} if and only if $\varphi$ is 1-homogeneous, increasing, and concave.
\end{proposition}
\begin{proof}
  We have already shown in Lemma~\ref{lem:construction-homogeneous} and the proof of Proposition~\ref{prop:construction-5-axioms} that any potential resulting from the construction is 1-homogeneous, increasing, and concave.
  Now let $\varphi$ be a given 1-homogeneous, increasing, and concave potential.
  As in Theorem~\ref{thm:CFMM-to-cost}, define
  $C(\q) = \sup\left\{ c \in \reals \mid \varphi(-\q - c\ones) \geq 1 \right\}$.
  Then for any $\q \succ \0$, we have $C(-\q) = 0 \iff \varphi(\q-0\ones) = 1 \iff \varphi(\q) = 1$.
  Now apply Construction~\ref{cons:bounded-reserves} to $C$ and obtain $\varphi'$.
  Then $C(-\q) = 0 \iff C(-\q/1)=0 \iff \varphi'(\q) = 1$.
  We conclude $\varphi(\q) = 1 \iff C(-\q) = 0 \iff \varphi'(\q) = 1$.
  As both $\varphi$ and $\varphi'$ are 1-homogeneous, they are fully determined by any level set, so must now be equal; formally, for any $\q\succ\0$, letting $\alpha = \varphi(\q)$, we have $\varphi(\q/\alpha) = \varphi(\q)/\varphi(\q) = 1 \implies \varphi'(\q/\alpha) = 1 \implies \varphi'(\q) = \alpha\varphi'(\q/\alpha) = \alpha = \varphi(\q)$.
\end{proof}

As a result of Proposition~\ref{prop:construction-1-homog}, through Corollary~\ref{cor:CFMM-scoring}, every 1-homogeneous CFMM can be associated with a unique proper scoring rule; the various level sets of the potential merely correspond to scaling up the underlying scoring rule.
As we will see in \S~\ref{sec:bounded-res-examples}, this fact can be used to construct useful CFMMs from the vast literature on proper scoring rules.

\begin{corollary}
  \label{cor:unique-scoring-rule}
  Every CFMM for a 1-homogeneous, concave, increasing $\varphi$ is associated with a unique proper scoring rule. 
\end{corollary}

Before moving to examples, let us briefly remark on connections to convex geometry.
Construction~\ref{cons:bounded-reserves} can be regarded geometrically as a gauge function (also known as a Minkowski functional) on a star set, which is known to be a characterization of 1-homogeneous functions.
For any set $K \subseteq \reals^n$, the gauge function (Minkowski functional) of $K$ is defined as $g_K(\q) = \inf\{c>0 \mid \q \in cK\}$.
Letting $K = \{\q \in \reals^n_{>0} \mid C(-\q) \leq 0\}$, we have $\varphi(\q) = \inf\{c>0 \mid \q/c \in K\} = g_K(\q)$.

\subsection{Examples}
\label{sec:bounded-res-examples}

\paragraph{Uniswap}

Let us perform Construction~\ref{cons:bounded-reserves} on the cost function $C_k$ we derived from Uniswap, $\varphi_U(\q) = \sqrt{q_1 q_2}$, given some level set $k > 0$:
\begin{align*}
  C_k(\q) = \frac 1 2 \left( q_1 + q_2 + \sqrt{4k^2 + (q_1-q_2)^2} \right)~.
\end{align*}
First, let us check that $C_k$ satisfies the condition of Construction~\ref{cons:bounded-reserves}: indeed, $C_k(\0) = k > 0$.
From \S~\ref{sec:equivalence-examples}, we have $C_k(-\q) = 0 \iff \varphi_U(\q) = \sqrt{q_1q_2} = k$.
Thus $C(-\q/\alpha) = 0 \iff \sqrt{q_1q_2} = \alpha k$.
Thus, the construction yields $\varphi_k(\q) = \sqrt{q_1 q_2}/k$, which is exactly $\varphi_U$ when $k=1$, and a scaled version for other values.

\paragraph{LMSR}
Let $C(\q) = b\log(\sum_{i=1}^n e^{q_i/b})$ be LMSR.
Unlike our choice $\varphi(\q) = -C(-\q)$ from \S~\ref{sec:equivalence-examples}, let us now consider a ``reserves-aware'' version, which will automatically scale liquidity with $\q_0$.
As $C(\0) = b \log n > 0$, Construction~\ref{cons:bounded-reserves} applies, giving a potential $\varphi$.
See Figure~\ref{fig:perspective-construction}(M) for a visualization.
Moreover, the condition of Proposition~\ref{prop:construction-5-axioms} is also satisfied, as $C(\q) = 0 \implies \sum_{i=1}^n e^{q_i} = 1 \implies q_i < 0$ for all $i$.
So the CFMM with potential $\varphi$ satisfies all five of our axioms.
Unfortunately, unlike Uniswap, this $\varphi$ does not permit a closed-form expression.
Nonetheless, as we discuss in \S~\ref{sec:future}, we believe it could be practical to deploy in real markets.

\paragraph{CFMM from Brier score}

Let us see how to take one of the most popular scoring rules, the Brier (or quadratic) score, \begin{equation}
  \label{eq:brier}
  S(\p,y_i) = 2\p\cdot\bm{\delta}_i - \|\p\|_2^2~,
\end{equation}
and convert it to a CFMM using our chain of equivalences.
The corresponding convex generating function is $G(\p) = \|\p\|^2$.
The associated cost function $C$, the convex conjugate of $G$, does not have a convenient closed form in general; for $n=2$ assets it can be written
\begin{equation}
  \label{eq:6}
  C(\q) =
  \begin{cases}
    q_1 & q_1-q_2 \geq 2\\
    q_2 & q_1-q_2 \leq -2\\
    \tfrac 1 8 \left((q_1-q_2)^2 + 4 (1 + q_1 + q_2)\right) & \text{otherwise}
  \end{cases}~.
\end{equation}
From this form, one can see that $C$ is not always increasing, as e.g.\ the price for asset 2 becomes 0 when $q_1 - q_2 \geq 2$.
The resulting market therefore does not satisfy \nodominatedtrades, as selling asset 2 for a price of 0 is dominated by doing nothing.

It turns out, however, that we may still apply Construction~\ref{cons:bounded-reserves} to obtain a sensible CFMM; in this case, for $n=2$, we obtain
\begin{equation}
  \label{eq:4}
  \varphi(\q) = \sqrt{q_1 q_2} + \frac {q_1 + q_2} 2~,
\end{equation}
the sum of the geometric and arithmetic means; see Figure~\ref{fig:perspective-construction}(R).
In fact, this potential function is a hybrid CFMM appearing already in the literature, e.g. \citet[\S~2.4]{angeris2022constant}.

\paragraph{Curve}
Finally, let us consider Curve~\citep{egorov2019stableswap}, given by $\varphi(\q) = \sum_i q_i + \sum_i (1/q_i)$.
As noted by \citet{angeris2022constant}, this potential is not 1-homogeneous, and thus could not be produced by Construction~\ref{cons:bounded-reserves}.
As a result, its level sets are not merely scaled copies of each other, but change as the value of $\varphi$ grows.
(For instance, its 0-level set is the same as Uniswap with $k=1$, but no other level set is in common with Uniswap.)
Nonetheless, we may still apply Theorem~\ref{thm:cfmm-char-boundedreserves} to see that $\varphi$ satisfies all five of our axioms; to check condition (c), note that $\bar\varphi = -\infty$ on $\bdposorth$.

\section{Adaptive Liquidity and Other Future Directions}
\label{sec:future}

The literature on automated market makers for decentralized exchanges has largely proceeded somewhat removed from the expansive literature on prediction markets.
We believe the results presented here will allow for these independent lines of research to merge and inform each other.
To conclude, we briefly discuss several avenues for future work, with a focus on liquidity adaptation.

\paragraph{Transaction fees}
In practice, CFMMs often allow liquidity to change over time, in two ways: (1) they may charge a ``transaction fee'', wherein traders must contribute directly to the reserves in addition to their trade, and (2) they may allow liquidity providers to contribute to the reserves in exchange for a dividend.
Let us first discuss the transaction fee.

Here the market designer chooses a parameter $\gamma \in (0,1]$, where lower $\gamma$ corresponds to a higher fee.
The market maker then accepts any trade keeping the potential function $\varphi$ constant after discounting the bundle given to the market maker by $\gamma$.
Formally, we can write
\begin{align}
  \label{eq:cfmm-transaction-fee-original}
  \valtrades_{\varphi,\gamma}(h) = \{\r\in\reals^n \mid \varphi(\hsum(h) + \gamma \r_+ - \r_-) = \varphi(\hsum(h))\}~,
\end{align}
where $\r_+ = \max(\r,\0)$ and $\r_- = \min(\r,\0)$.
To state this set of valid trades more naturally in terms of transaction fees, relative to the no-fee $\valtrades$ for the vanilla CFMM, we may write
\begin{align}
  \label{eq:cfmm-transaction-fee-reformulated}
  \valtrades_{\varphi,\gamma}(h) = \{\r + \fee(\r) \mid \r\in\valtrades_\varphi(h) \}~,
\end{align}
where $\fee(\r) = \beta \r_+$ for $\beta = (1-\gamma)/\gamma > 0$.
That is, a trader may choose any trade $\r$ keeping $\varphi$ constant, but then must also add $\beta \r_+$ to the reserves.

Now suppose $\varphi$ is increasing, concave, and 1-homogeneous, as is commonly the case, and as guaranteed by Construction~\ref{cons:bounded-reserves}.
Then as $\varphi(\q_h) = \varphi(\q_h + \r)$, and $\varphi$ is increasing, we will have $\varphi(\q_h + \r + \fee(\r)) > \varphi(\q_h + \r) = \varphi(\q_h)$.
As we saw in \S~\ref{sec:reserves-and-liqudity}, $\varphi(\q)$ represents the liquidity level when the reserves are $\q$, so we conclude that the liquidity increases after each trade.
Moreover, as the fee-less CFMM satisfies \boundedreserves, and $\varphi$ is increasing, it is easy to verify that \boundedreserves will still be satisfied with the transaction fee.
Thus, the transaction fee successsfully subsidizes the liquidity increase of the market without risking depleting the reserves.

\paragraph{Implicitly defined potential functions}
Recall that the $\varphi$ from Construnction~\ref{cons:bounded-reserves} is only implicitly defined, as the solution $\varphi(\alpha)$ to $C(-\q/\alpha) = 0$ for the given cost function $C$.
Thus, even if $C$ is given explicitly, $\varphi$ may not have a closed form.
While sometimes one can solve for $\varphi$ explicitly, as we saw in \S~\ref{sec:bounded-res-examples}, this is the case for the 1-homogeneous CFMM potential $\varphi$ we derived from LMSR.
The lack of a closed form poses a challenge, as the transaction fee causes $\varphi$ to change, and thus the value of $\varphi(\q_h)$ would need to be recalculated rather than being fixed ahead of time.

We now propose a straightforward workaround using the fact that $\varphi$ is still implicitly defined by a known cost function $C$.
Suppose the current value $\alpha = \varphi(\q_h)$ is publicly known, and a trader wishes to purchase $\r \in \valtrades_\varphi(\q_h)$, i.e., such that $\varphi(\q_h + \r) = \alpha$.
We will ask the trader to announce $\r$, as well as the value $\alpha' = \varphi(\q_h + \r + \fee(\r))$, up to some error tolerance.
As $C(-\q/\alpha)$ is monotone in $\alpha$, the trader can easily compute $\alpha'$ within the desired accuracy given an expression for $C$.
Moreover, the two relevant conditions can be checked on-chain: $\varphi(\q_h + \r) = \alpha$ by $C(-(\q_h+\r)/\alpha) = 0$, and $\varphi(\q_h + \r + \fee(\r)) = \alpha'$ by $C(-(\q_h+\r+\fee(\r))/\alpha') = 0$.

To illustrate, let $C$ be LMSR, and $\varphi$ the result of Construction~\ref{cons:bounded-reserves}.
From the condition $C(-\q/\alpha) = 0$, the level set $\varphi(\q) = \alpha$ is $\{\q \in \reals^n_{>0} \mid \sum_{i=1}^n e^{-q_i/\alpha} = 1\}$.
The validity of $\r$ and $\alpha'$ can be checked, via $\sum_{i=1}^n \exp(-(q_i+r_i)/\alpha) \approx 1$ and $\sum_{i=1}^n \exp(-(q_i+r_i +\fee(\r)_i)/\alpha') \approx 1$.
Thus, by asking a trader to compute a valid trade, as well as approximating the next value of $\varphi$, the market can proceed even without a closed form for $\varphi$.

\paragraph{Other forms of liquidity adaptation}
The prediction market literature offers many other forms of transaction fees and schemes for liquidity adaptation; see e.g. \citet{othman2010practical,othman2012profit-charging,li2013axiomatic}.
Using the equivalence developed in this work, these market makers can be readily transformed into those for decentralized exchanges.
It would be particularly interesting to study the CFMM transaction fee within the volume-parameterized market (VPM) framework of~\citet{abernethy2014general}; we conjecture that it fails to satisfy their \emph{shrinking spread} axiom (for any choice of volume function) but satisfies their others.
Conversely, the perspective market from that same paper may be of interest to the decentralized finance community.

\paragraph{Other directions}
As mentioned above, many CFMMs allow liquidity providers to directly contribute to the reserves.
In light of our results, it would be interesting and potentially impactful to study the elicitation implications of these protocols from the side of liquidity providers.
Finally, while the above focuses on potential contributions from the prediction market literature to decentralized finance, contributions in the direction are likely to be fruitful as well.

\subsection*{Acknowledgements}
We thank David Pennock, Daniel Reeves, and Anson Kahng for collaboration in working out the cost function and proper scoring rule corresponding to Uniswap.
We also thank Scott Kominers, Ciamac Moallemi, Abe Othman, Tim Roughgarden, and Christoph Schlegel for helpful discussions.
This material is based upon work supported by the National Science Foundation under Grant No.\ IIS-2045347.

\bibliographystyle{plainnat}
\bibliography{diss,cfmms,extra}

\appendix

\section{Omitted Proofs for CFMM Characterization}
\label{sec:cfmm-lemmas}

To prove Lemma \ref{lemma:fixedq-DR-Rplusconvex}, we separate out the following lemma, which will be useful later.
\begin{lemma} \label{lemma:axioms-superlevel-convex}
  Suppose the CFMM defined by $\varphi$ satisfies \pathindependence, \nodominatedtrades, \liquidation, and \demandresponsiveness.
  For any $c \in \reals$, let $\R := \{\q \mid \varphi(\q) = c\}$ and let $\R^+ := \{\q + \r \mid \q \in \R, \r \preceq \0\}$.
  Then $\R^+$ is convex.
\end{lemma}
\begin{proof}
  We show that for any $\q,\q' \in \R$, any convex combination $\q^* = (1-\lambda)\q + \lambda \q'$ is in $\R^+$.
  The lemma then follows because for any $\r,\r' \in \R^+$ of the form $\q+\mathbf{z},\q'+\mathbf{z}'$, with $\mathbf{z},\mathbf{z}' \geq \0$, the convex combination is $(1-\lambda)\r + \lambda \r' = \q^* + (1-\lambda)\mathbf{z} + \lambda \mathbf{z}'$, which is in $\R^+$.

  So let $\q,\q' \in \R$, let $\lambda \in (0,1)$, and define $\r = \q' - \q$.
  Then $\r \in \valtrades(\q)$.
  Let $\r^+ = \max(\r, \vec{0})$ and $\r^- = -\min(\r, \vec{0})$ taken pointwise, so that $\r = \r^+ - \r^-$.
  By \liquidation, there exists $\beta \geq 0$ such that $\r_1 := \lambda \r^+ - \beta \r^- \in \valtrades(\q)$.
  By \nodominatedtrades, $\beta \neq 0$ as otherwise $\lambda \r^+$ would be dominated by trade $\0$.
  Let $\r_2 = (1-\lambda)\r^+ - (1-\beta)\r^-$.
  Then $\r_2 \in \valtrades(\q + \r_1)$, because $\q + \r_1 + \r_2 = \q + \r = \q'$.
  Therefore, by \demandresponsiveness, $\frac{1-\beta}{\beta} \leq \frac{1-\lambda}{\lambda}$, implying $\lambda \leq \beta$.

  Observe that $\q^* = \q + \lambda \r = \q + \lambda \r^+ - \lambda \r^- = \q + \r_1 + (\beta - \lambda) \r^-$.
  Because $\q + \r_1 \in \R$, $\beta \geq \lambda$, and $\r^- \succeq \vec{0}$, we have $\q^* \in \R^+$.
\end{proof}

\begin{proof}[Proof of Lemma \ref{lemma:fixedq-DR-Rplusconvex}]
  By Lemma \ref{lemma:fixedq-axioms-cfmm}, the market can be implemented as a CFMM with potential $\varphi$.
  Applying Lemma \ref{lemma:axioms-superlevel-convex} to $c = \varphi(\q_0)$ in particular, we obtain $\R^+_{\varphi}(\q_0)$ is convex, as claimed.
\end{proof}

Proposition \ref{prop:fixedq-cfmm} will be proven by the following two lemmas.
\begin{lemma} \label{lemma:fixedq-Rplusconvex-cfmm}
  Fix the initial reserves $\q_0$.
  If a market maker $\valtrades$ with $\valhist(\epsilon) = \{(\q_0)\}$ satisfies \liquidation, \nodominatedtrades, \pathindependence, and \demandresponsiveness, then it can be implemented as a CFMM with an increasing, concave potential function $\varphi$.
\end{lemma}
\begin{proof}
  Given a market $\valtrades$ satisfying the axioms, Lemma \ref{lemma:fixedq-DR-Rplusconvex} implies it has a potential $\varphi'$ such that $\R^+_{\varphi'}(\q_0)$ is convex.
  Again, for convenience, write $\R = \R_{\varphi'}(\q_0)$ and $\R^+ = \R^+_{\varphi'}(\q_0)$.
  Define
    \[ \varphi(\q) = \sup \left\{ u \in \reals \mid \q - u\ones \in \R^+ \right\}.  \]
    This is a well-known construction in theory of risk measures \citep{follmer2008convex,follmer2015axiomatic}; we provide a proof that $\varphi$ is increasing and concave for completeness.
  First, $\varphi(\q)$ is well-defined, i.e. the supremum is over a nonempty set: for example, taking $u = \min_i (\q - \qZ)_i$, we find that $\q - u\ones - \qZ \succeq \0$, i.e. $\q - u\ones \succeq \qZ$, so $\q - u\ones \in \R^+$.

  Second, it implements $\valtrades$ as a CFMM.
  By definition of the set of legal reserves $\R$, we have $\r \in \valtrades(h) \iff \qh + \r \in \R$.
  So we must show $\varphi(\q) = \varphi(\q_0) \iff \q \in \R$.
  First, for any $\q \in \R$, including $\q=\q_0$, we have $\varphi(\q) = 0$.
  This follows because $\q \in \R \implies \varphi(\q) \geq 0$, but $\varphi(\q) > 0$ would violate \nodominatedtrades, as we would have $\q - \varphi(\q)\ones \in \R^+$, which implies some trade in $\R$ strictly dominates $\0$.
  Second, if $\varphi(\q) = 0$, then $\q \in \R^+$, and $\q$ cannot dominate any member of $\R$ (else $\varphi(\q) > 0$), so it must be in $\R$.
  We have shown $\q \in \R \iff \varphi(\q) = 0 = \varphi(\q_0)$.
  
  Now we show $\varphi$ is increasing.
  Let $\q' \succeq \q$ and $\r = \q' - \q$.
  Let $u \in \reals$ such that $\R^+$ contains $\q - u\ones$.
  Then $\R^+$ also contains, by upward closure, $\q - u\ones + \r = \q' - u\ones$.
  This proves $\varphi(\q') \geq \varphi(\q)$, as the supremum is taken over a superset.
  Furthermore, if $q' \neq q$, then we claim $\varphi(\q') > \varphi(\q)$, i.e. $\varphi$ is increasing.
  If $\varphi(\q') = \varphi(\q) = u$, then $\q - u\ones, \q' - u\ones \in \R$, but then the trade $\q-u\ones-\qZ$ is dominated by the trade $\q'-u\ones-\qZ$, both in $\valtrades((\qZ))$, contradicting \nodominatedtrades.

  Finally, for concavity, let $\lambda \in [0,1]$ and $\q^* = (1-\lambda)\q + \lambda \q'$.
  Let $\q-u\ones,\q'-u'\ones \in \R^+$.
  Then $\R^+$ also contains $(1-\lambda)(\q-u\ones) + \lambda(\q'-u'\ones) = \q^* - [(1-\lambda)u + \lambda u']\ones$, proving that $\varphi(\q^*) \geq (1-\lambda)\varphi(\q) + \lambda \varphi(\q')$.
\end{proof}

\begin{lemma} \label{lemma:fixedq-cfmm-DR}
  Fix the initial reserves $\q_0$.
  If a market maker $\valtrades$ with $\valhist(\epsilon) = \{(\q_0)\}$ is a CFMM with an increasing, concave potential function $\varphi$, then it satisfies \liquidation, \nodominatedtrades, \pathindependence, and \demandresponsiveness.
\end{lemma}
\begin{proof}
  Let $h\in\valhist(\epsilon)$ and, for short, let $\q = \qh$.

  (\demandresponsiveness)
  Assume $\r = \r_+-\r_-\in \valtrades(h)$, where $\r_+,\r_- \succeq \0$, and assume $\alpha \r_+-\beta \r_- \in \valtrades(h \oplus \r)$ where $\alpha,\beta>0$.
  We must prove $\alpha \leq \beta$.
  As these are feasible trades, we have that the value of the CFMM is same at those points i.e. $\varphi(\q)=\varphi(\q+\r)=\varphi(\q+(1+\alpha)\r_+-(1+\beta)\r_-)$.
  Let $\lambda = \frac{1}{1+\alpha}$.
  By concavity of $\varphi$ we have that 
  \begin{align*}
    \varphi\Bigl(\q+\r_+-\frac{1+\beta}{1+\alpha}\r_-\Bigr)
    &= \varphi\Bigl((1-\lambda)\q+\lambda(\q+(1+\alpha)\r_+-(1+\beta)\r_-)\Bigr)
    \\
    &\geq (1 - \lambda) \varphi(\q) + \lambda \varphi(\q + (1+\alpha)\r_+ - (1+\beta)\r_-)
    \\    
    &= \lambda\varphi(\q+\r_+-\r_-) +(1-\lambda)\varphi(\q+\r_+-\r_-)\\
    &= \varphi(\q+\r_+-\r_-) .
  \end{align*}
  As $\varphi$ is an increasing function, this implies we cannot have $\q+\r_+-\frac{1+\beta}{1+\alpha}\r_- \precneqq \q+\r_+-\r_-$.
  Thus we cannot have $\frac{1+\beta}{1+\alpha}\r_- \succneqq \r_-$.
  We conclude $\frac{1+\beta}{1+\alpha}\leq 1$, and thus $\alpha \geq \beta$, which proves \demandresponsiveness.

  \pathindependence is immediately satisfied by a CFMM, as $\r \in \valtrades(h)$ and $\r' \in \valtrades(h \oplus \r)$ imply $\varphi(\q) = \varphi(\q+\r) = \varphi(\q+\r+\r')$, which implies $\r+\r' \in \valtrades(h)$.
  
  Assume \nodominatedtrades is not satisfied i.e. $\r,\r' \in \valtrades(h)$ with $\r' \succneqq \r$.
  This implies that $\varphi(\q)=\varphi(\q + \r)=\varphi(\q +\r')$, violating that $\varphi$ is increasing.

  For \liquidation, let $\r,\r'\succneqq \0$.
  Concavity of $\varphi$ implies that at current reserves $\q$, there exists a supergradient%
  \footnote{A supergradient is defined as any $\mathbf{d}$ that satisfies the inequality $\mathbf{d}\cdot(\q'-\q)\geq \varphi(\q')-\varphi(\q), \forall \q'$.}
  of $\varphi$, say $\mathbf{d}$.
  We show that as $\varphi$ is increasing, $\mathbf{d} \succ \0$.
  This should hold for $\q'=\q+\bm{\delta}_i , \forall i$.
  Since $\varphi$ is increasing, $\varphi(\q') >\varphi(\q)$.
  Therefore, $\mathbf{d}\cdot \bm{\delta}_i = d_i > 0$ and $\mathbf{d}\succ \0$.
  
  For any $\beta' > \frac{\mathbf{d}\cdot \r}{\mathbf{d}\cdot \r'}$, $\mathbf{d}\cdot(\r-\beta' \r') < 0$.
  Hence from the supergradient inequality, $\varphi(\q)+\mathbf{d}\cdot(\r-\beta' \r') \geq \varphi(\q + \r -\beta' \r')$, from which we have $\varphi(\q) > \varphi(\q + \r -\beta' \r')$.
  As $\varphi$ is increasing, we also have that $\varphi(\q) \leq \varphi(\q+\r-0\cdot \r')$.
  Hence by continuity of $\varphi$, which follows from concavity, we can say that $\exists \beta \in [0,\beta']$ s.t. $\varphi(\q)=\varphi(\q + \r -\beta \r')$ which proves \liquidation.

\end{proof}

Theorem \ref{thm:itsacfmm} will be proven by the following two lemmas.
We show Lemma \ref{lemma:Rplusconvex-cfmm} on a slightly more general space called \emph{dominant domain}, which we define below.
The two primary examples of dominant domains are $\reals^n$ itself and the strictly positive orthant $\reals^n_{>0}$.

\begin{definition} \label{def:dominant-domain}
  We call a subset $S$ of $\reals^n$ a \emph{dominant domain} if it is an open, convex set that is closed under dominance, i.e. $\r \in S, \r' \succeq \r \implies \r' \in S$.
\end{definition}

\begin{lemma} \label{lemma:Rplusconvex-cfmm}
  If a market maker, where $\valtrades(\epsilon)$ is a dominant domain $S$, satisfies \liquidation, \nodominatedtrades, \strongpathindependence, and \demandresponsiveness,
then it is implemented by a CFMM for some increasing, continuous, quasiconcave $\varphi:S\to\reals$.
\end{lemma}
\begin{proof}
  Let a market $\valtrades$ satisfying the axioms be given.
  For any $h \in \valhist(\epsilon)$, define
    \[ \varphi(\qh) = \alpha \in \reals \text{ such that } \alpha\ones - \qh \in \valtrades(h)~.  \]
  If there is no $h$ such that $\q = \qh$, then leave $\varphi(\q)$ undefined.

  We first show that $\varphi$ is a well-defined function on $S$.
  Let $h \in \valhist(\epsilon)$ and $\q = \qh$.
  Observe $\q \in S$ by \strongpathindependence.
  Let $c = \max_i \q_i + 1$ and let $\r = c \ones - \q$, noting that $\r \succ \0$.
  By \liquidation applied to $\r$ and $\r' = \ones$, there exists $\beta \geq 0$ such that $\valtrades(h)$ contains $\r - \beta \ones = (c - \beta)\ones - \q$.
  We take $\alpha = c - \beta$.
  Furthermore, $\alpha$ is unique, as $\beta$ is unique by \nodominatedtrades.
  
  We will show that $\varphi:S\to\reals$ implements $\valtrades$ as a CFMM.
  Then, we will show that it is increasing, continuous, and quasiconcave.

  To begin, we first show that $\valtrades$ has ``level set structure''.
  Specifically, we will show that $\q \equiv \q' \iff \q'-\q \in \valtrades((\q))$ is a well-defined equivalence relation.
  Combining this statement with the fact that $\varphi$ is well-defined, we will have $\q \equiv \varphi(\q)\ones$ for all $\q \in S$.

  As \strongpathindependence implies \pathindependence, results only assuming the latter apply.
  For all $\q \in S$, let $\R(\q)$ be the set from Lemma~\ref{lemma:fixedq-axioms-cfmm}, so that $\r \in \valtrades(h) \;\&\; \q_0=\q \iff \q_h + \r \in \R(\q)$.
  Substituting $h = (\q), \r = \q' - \q$ into the above, the following claim shows that $\equiv$ as defined above is indeed an equivalence relation.

  \begin{claim}\label{claim:reserves-closed}
    For all $\q,\q'\in S$, $\q \in \R(\q') \iff \R(\q) = \R(\q') \iff \q' \in \R(\q)$.
  \end{claim}
  \begin{proof}
    We first show $\q \in \R(\q') \implies \R(\q) = \R(\q')$.
    Let $\q,\q'$ such that $\q \in \R(\q')$.
    Let $\q''\in\valtrades(\epsilon)$.
    Let $h=(\q), h'=(\q',\q-\q'), \r = \q''-\q$.
    Note $h' \in \valhist(\epsilon)$ since $\q\in\R(\q')$ by assumption and thus $\q-\q' \in \valtrades((\q'))$.
    Observe $\q_h = \q_{h'}$ so \strongpathindependence applies.
    From Lemma~\ref{lemma:fixedq-axioms-cfmm} and \strongpathindependence, we have $\q'' \in \R(\q) \iff \q_h + \r \in \R(\q) \iff \r \in \valtrades(h) \iff \r \in \valtrades(h') \iff \q_{h'} + \r \in \R(\q') \iff \q'' \in \R(\q')$.
    For the reverse implication, if $\R(\q) = \R(\q')$, we have $\0 \in \valtrades((\q')) \implies \q' \in \R(\q') = \R(\q)$, completing the proof.
  \end{proof}

  We can now see that $\varphi$ implements $\valtrades$.
  Let $h\in\valhist(\epsilon)$.
  Then from \strongpathindependence, $\r \in \valtrades(h) \iff \r \in \valtrades((\q_h)) \iff \q_h \equiv \q_h + \r \iff \varphi(\q_h)\ones \equiv \varphi(\q_h+\r)\ones \iff \varphi(\q_h) = \varphi(\q_h+\r)$. It remains to show that $\varphi$ is continuous, increasing, and quasiconcave.

  Before we show properties of $\varphi$, we introduce another construction of a ``cost function''.
  Let $C_{\q}(\r) = \alpha$ such that $\r-\alpha\ones\in\R(\q)$.
  Observe that this construction is ones-invariant i.e. $C_{\q}(\r + a\ones) = C_{\q}(\r) + a$ for any constant $a$.
  We first show that this $C_{\q}(\r)$ is continuous for any $\r$.

  Suppose that $C_{\q}$ is not continuous at a particular $\r$.
  Thus $\exists \epsilon >0$ such that $\forall \delta >0$, $\exists \r'$ such that $||\r-\r'||_{\infty} < \delta$ and $||C_{\q}(\r)-C_{\q}(\r')||_{\infty} \geq \epsilon$.
  Given an $\epsilon > 0 $, take $\delta = \epsilon/2$, and let $\r'$ be given satisfying the above.
  Let $\alpha = C_{\q}(\r)$, $\beta = C_{\q}(\r')$ so that $\r-\alpha\ones\in \R(\q)$ and $\r'-\beta\ones\in \R(\q)$.
  As $||\r-\r'||_{\infty} < \delta$, we have $ -\delta\ones \prec \r-\r' \prec \delta\ones$.
  Assume without loss of generality that $\beta > \alpha$; then $\beta-\alpha \geq \epsilon$.
  Now $(\r-\alpha\ones)-(\r'-\beta\ones) = (\r-\r') - (\beta-\alpha)\ones \succ -\delta\ones + \epsilon\ones \succ \0$.
  We have now contradicted \nodominatedtrades as $\r-\alpha\ones \succ \r'-\beta\ones$ but $\r-\alpha\ones,\r'-\beta\ones \in \R(\q) $.

  (Increasing)
  Suppose for a contradiction that we have two vectors $\q,\q'$ such that $\q'\succneqq \q$ and $\varphi(\q)\geq\varphi(\q')$.
  If $\varphi(\q) = \varphi(\q')$ then $\q' - \q \in \valtrades((\q))$, violating \nodominatedtrades as $\q' - \q \succneqq \0$.
  Otherwise, $\varphi(\q) > \varphi(\q')$.
  As we saw from above that $C_{\q},C_{\q'}$ are continuous functions, so is the function $f(\r) = C_{\q'}(\r)-C_{\q}(\r)$.
  Observe that $C_{\q'}(\varphi(\q)\ones) = \varphi(\q)-\varphi(\q')$ as $\varphi(\q')\ones\in \R(\q')$.
  Thus, $f(\varphi(\q)\ones) = \varphi(\q)-\varphi(\q') - 0 > 0$.
  On the other hand, $f(\q)= C_{\q'}(\q) - 0$.
  If $\alpha := C_{\q'}(\q) \geq 0$, then $\q' \succneqq \q-\alpha\ones$ but $\q-\alpha\ones\in\R(\q')$ by definition of $C_{\q'}$, violating \nodominatedtrades.
  Thus, $f(\q) = C_{\q'}(\q) < 0$.
  As $f$ is continuous, by the intermediate value theorem we have some $\r\in S$ such that $f(\r) = 0$, implying $C_{\q'}(\r)=C_{\q}(\r) =: \gamma$.
  From the definition of $C_\q$ and $C_{\q'}$, we have $\r-\gamma\ones\in \R(\q)$ and $\r-\gamma\ones\in \R(\q')$.
  From Claim \ref{claim:reserves-closed}, $\q,\q'\in \R(\r-\gamma\ones)$ which is contradiction of \nodominatedtrades.
  
  (Quasiconcave)
  We claim that the superlevel sets of $\varphi$ are exactly the sets of the form $\{\q + \r \mid \varphi(\q) = c, \r \succeq \0\}$ for some $c$.
  Then Lemma \ref{lemma:axioms-superlevel-convex} implies that each such set is convex, which proves quasiconcavity.
  To show these are the superlevel sets: for one direction, because $\varphi$ is increasing, we have $\varphi(\q) = c, \r \succeq \0 \implies \varphi(\q + \r) \geq c$.
  For the other direction, suppose $\varphi(\q') \geq \varphi(\q)$.
  We show $\q' = \q'' + \r$ for some $\varphi(\q'') = c, \r \succeq \0$.
  Specifically, use \liquidation starting at reserves $\qh := \q$ to trade $\q' - \q + \gamma \ones$ for units of $\ones$, where $\gamma$ is chosen large enough that the trade vector is positive.
  Receiving $\beta \ones$ for some $\beta > 0$, the new reserves are $\q'' = \q' + (\gamma - \beta)\ones$.
  By definition of a CFMM, we have $\varphi(\q'') = \varphi(\qh) = c$.
  And because $\varphi(\q') \geq c$, and $\varphi$ is increasing, we must have $\q'' \preceq \q'$.

  (Continuous)
  Let $\q \in S$.
  We will show that $\varphi$ is continuous at $\q$.
  By \citet[Proposition 4.1]{crouzeix2005continuity}, as $\varphi$ is increasing with respect to $\reals^n_{>0}$, and $\ones$ is in the interior of $\reals^n_{>0}$, it suffices to show continuity of the function $f:\reals\to\reals, \gamma \mapsto \varphi(\q+\gamma\ones)$ at $0$.
  Observe that $f$ is increasing.
  We use the $\epsilon,\delta$ formulation of continuity.

  Let $\alpha = \varphi(\q)$ and let $\epsilon>0$ such that $(\alpha\pm\epsilon)\ones \in S$, which exists as $S$ is open.
  Let $\q_0 = (\alpha+\epsilon)\ones \in S$; by definition, $\varphi(\q_0) = \alpha+\epsilon$.
  Let $c > \min_i q_i$ so that $\r := \q + c\ones \succ \0$.
  Applying \liquidation with $h = (\q_0)$ and bundles $\r$ and $\r'=\ones$, we have some $\beta \geq 0$ such that $(\q+c\ones) - \beta \ones \in \valtrades((\q_0))$.
  This condition is in turn equivalent to $\varphi(\q + c\ones - \beta\ones) = \varphi(\q_0) = \alpha + \epsilon$.
  By the same argument for $\q_0 = (\alpha - \epsilon)\ones$, we have some $\beta'\geq 0$ such that $\varphi(\q + c\ones - \beta'\ones) = \alpha - \epsilon$.

  We have $f(c-\beta') = \alpha - \epsilon$, $f(0) = \alpha$, and $f(c-\beta) = \alpha + \epsilon$.
  As $f$ is increasing, we must have $c-\beta' < 0 < c-\beta$.
  Let $\delta = \min(|c-\beta'|,|c-\beta|) > 0$.
  Then $c-\beta' \leq -\delta < 0 < \delta \leq c-\beta$.
  In particular, again as $f$ is increasing, for all $\gamma \in (-\delta,\delta)$, we have $|f(\gamma) - f(0)| < \max\{f(0) - f(-\delta), f(\delta) - f(0)\} \leq  \max\{f(0) - f(c-\beta'), f(c-\beta) - f(0)\} = \epsilon$.
  Thus, $|f(0) - f(\gamma)| < \epsilon$ whenever $|0 - \gamma| < \delta$, giving continuity of $f$ at $0$ and completing the proof.  
\end{proof}

\begin{lemma} \label{lem:cfmm-axioms}
  Let a CFMM be given for a quasiconcave, increasing, continuous $\varphi:\reals^n\to\reals$, with $\valtrades(\epsilon) = \reals^n$.
  Then it satisfies \demandresponsiveness, \liquidation, \nodominatedtrades and \pathindependence.
\end{lemma}

\begin{proof}
Let $h\in\valhist(\epsilon)$ and, for short, let $\q = \qh$.

  (\demandresponsiveness)
  Assume $\r = \r_+-\r_-\in \valtrades(h)$, where $\r_+,\r_- \succeq \0$, and assume $\alpha \r_+-\beta \r_- \in \valtrades(h \oplus \r)$ where $\alpha,\beta>0$.
  We must prove $\alpha \leq \beta$.
  As these are feasible trades, we have that the value of the CFMM is same at those points i.e. $\varphi(\q)=\varphi(\q+\r)=\varphi(\q+(1+\alpha)\r_+-(1+\beta)\r_-)$.
  Let $\lambda = \frac{1}{1+\alpha}$.
  By quasiconcavity of $\varphi$ we have that 
  \begin{align*}
    \varphi\Bigl(\q+\r_+-\frac{1+\beta}{1+\alpha}\r_-\Bigr)
    &= \varphi\Bigl((1-\lambda)\q+\lambda(\q+(1+\alpha)\r_+-(1+\beta)\r_-)\Bigr)
    \\
    &\geq \min\{ \varphi(\q) ,  \varphi(\q + (1+\alpha)\r_+ - (1+\beta)\r_-)\}
    \\
    &= \min\{\varphi(\q+\r_+-\r_-) , \varphi(\q+\r_+-\r_-)\}\\
    &= \varphi(\q+\r_+-\r_-) .
  \end{align*}
  As $\varphi$ is an increasing function, this implies we cannot have $\q+\r_+-\frac{1+\beta}{1+\alpha}\r_- \precneqq \q+\r_+-\r_-$.
  Thus we cannot have $\frac{1+\beta}{1+\alpha}\r_- \succneqq \r_-$.
  We conclude $\frac{1+\beta}{1+\alpha}\leq 1$, and thus $\alpha \geq \beta$, which proves \demandresponsiveness.

  (\strongpathindependence) 
  $\forall h,h' \in \valhist(\epsilon)$, if  $\q = \hsum(h) = \hsum(h')$ then $\valtrades(h) = \{\r \mid \varphi(\q)=\varphi(\q +\r)\}$ which is also $\valtrades(h')$.
  We also trivially have $\q_h \in \reals^n = \valtrades(\epsilon)$ for all $h\in\valhist(\epsilon)$. 

  (\nodominatedtrades) Assume that the axiom is not satisfied, i.e., we have $\r,\r' \in \valtrades(h)$ with $\r' \succneqq \r$.
  This implies that $\varphi(\q)=\varphi(\q + \r)=\varphi(\q +\r')$, violating that $\varphi$ is increasing.

  (\liquidation)
  Quasi-concavity of $\varphi$ implies that at current reserves $\q$, there exists a supporting hyperplane%
  \footnote{A supporting hyperplane of a set $A$ would satisfy the inequality $ \mathbf{x}\cdot \mathbf{d} \geq  \q\cdot \mathbf{d}$, for all $\mathbf{x} \in A$ }
  $H = \{ \mathbf{x} \mid \mathbf{x}\cdot \mathbf{d} =  \q \cdot \mathbf{d}\}$ of super-level set%
  \footnote{ A super-level set of $\varphi$ for level $k$ is defined as $ U_{k} = \{\r \mid \varphi(\r) \geq \varphi(k)\}$}
  $U_{\q}$ of $\varphi$.
  The existence of a supporting hyperplane at $\q$ stems from $\q \in \partial U_{\q}$ as we will show now.
  We take a slight detour to show that $\partial U_{\q} = \{\r \mid \varphi(\r)=\varphi(\q)\}$.
  Lets assume the contrary that there exists a $\r \in \partial U_{\q}$ such that $\varphi(\r) > \varphi(\q)$.
  Let $\r' = \r - \epsilon\cdot \mathbf{d}$ for an $\epsilon > 0$ such that $\varphi(\q) < \varphi(\r') < \varphi(\r)$, which is possible as $\varphi$ is continuous. 
  As $\mathbf{d}\cdot \r' = \mathbf{d}\cdot \r -\epsilon \lVert \mathbf{d} \rVert^2 < \mathbf{d}\cdot\r$, this gives us $\r'\notin U_{\q}$ which is a contradiction as $\varphi(\r') > \varphi(\q)$ and $\r'$ is in the super-level set.
  
  The supporting hyperplane inequality should hold for $\q'=\q+\bm{\delta}_i$ for some $i$.
  Since $\varphi$ is increasing and defined on all of $\reals^n$, we have $\varphi(\q') >\varphi(\q)$ and hence $\q' \in U_{\q}$ and $\q' \notin \partial U_{\q}$.
  Therefore, $ \mathbf{d}\cdot \q' > \mathbf{d}\cdot \q \implies \mathbf{d}\cdot \bm{\delta}_i = d_i > 0$ and as it folds for any $i$, $\mathbf{d}\succ \0$.

  Let $\r,\r'\succneqq \0$ and let $\beta' > \frac{\mathbf{d}\cdot \r}{\mathbf{d}\cdot \r'}$ be a constant.
  Then $\mathbf{d}\cdot(\r-\beta' \r') < 0$ and $\mathbf{d}\cdot(\q+ \r-\beta' \r') < \mathbf{d}\cdot \q$.
  Hence from the supporting hyperplane inequality, $\q+ \r-\beta' \r' \notin U_{\q}$, from which we have $\varphi(\q) > \varphi(\q + \r -\beta' \r')$, again using the fact that $\varphi$ is defined at the latter point.
  As $\varphi$ is increasing, we also have that $\varphi(\q) \leq \varphi(\q+\r-0\cdot \r')$.
  Hence by continuity of $\varphi$, we can say that $\exists \beta \in [0,\beta']$ s.t. $\varphi(\q)=\varphi(\q + \r -\beta \r')$ which proves \liquidation.
\end{proof}

\section{Information Elicitation}
\label{app:elicitation}

\subsection{Sketch for Fact \ref{fact:elicitation-cost-function}}
\label{app:elicit-fact}

We recall that a \emph{scoring rule} is a function $S: \Delta_{\Y} \times \Y \to \reals \cup \{-\infty\}$, assigning a score $S(\p,y_i)$ to each prediction $\p$ when the obeserved outcome is $y_i$.
The scoring rule is \emph{proper} if the expected score, with respect to any given belief $\p^*$, is maximized by reporting $\p=\p^*$.
It is \emph{strictly proper} if $\p=\p^*$ uniquely maximizes the expected score.
Examples of strictly proper scoring rules are the log score $S(\p,y_i) = \log \p_i$ and the Brier or quadratic score $S(\p,y_i) = 2\p_i - \sum_j \p_j^2$.

We now sketch the proof of Fact \ref{fact:elicitation-cost-function} for completeness.

\begin{proof}[Proof sketch for Fact \ref{fact:elicitation-cost-function}]
  The first key step is Theorem 3.1 of \cite{frongillo2018axiomatic}, which states that any market maker satisfying Path Independence and Incentive Compatibility must in fact be representable as a ``scoring rule market''.
  A scoring rule market, in our terminology, is an automated market maker including an asset of cash where $\valtrades(h)$ must map one-to-one to predictions $\p_t$, and the net payoff for moving the market prediction from $\p_{t-1}$ to $\p_t$ must be given by the formula of \cite{hanson2003combinatorial}:
  \[ \text{net payoff} = S(\p_t,y_i) - S(\p_{t-1},y_i), \]
  where $S$ is a strictly proper scoring rule.
  The intuition is that Incentive Compatibility is characterized by use of a strictly proper scoring rule in each round, while Path Independence imposes the requirement of ``telescoping sums'' for sequences of predictions, giving the above formula.

  The second key step is given by classic results of \cite{chen2007utility,abernethy2013efficient} stating that scoring-rule markets are equivalent to cost-function markets.
  In particular, a strictly proper scoring rule $S$ corresponds via convex duality to a convex cost function $C$ such that
  \[ S(\p_t,y_i) - S(\p_{t-1},y_i) = (\q_t)_i - (\q_{t-1})_i - \left[ C(\q_t) - C(\q_{t-1}) \right] . \]
  (Also, $C$ is differentiable and has a certain set of gradients, namely the probability simplex.)
  Here convex duality gives a one-to-one correspondence between market states $\q_t$ and predictions $\p_t$, such that the trade $\r$ with $\q_t = \q_{t-1} + \r$ is equivalent to the prediction update $\p_{t-1} \to \p_t$.
  In other words, the cost-function interface to the market is equivalent to the scoring-rule interface.\footnote{There is a difference in timing, because as usually implemented, a cost function collects the cash payment at the time of trade while the securities pay off at the end; the scoring rule assesses all payments at the end. For Fact \ref{fact:elicitation-cost-function}, we must either wait to collect the payment until the end, or assume that one unit of cash has constant utility over time (the usual assumption).}

  To complete the sketch, we observe the properties of $C$.
  Convexity follows from the convex duality of scoring rules discussed above.
  Increasing and ones-invariance follow from Incentive Compatibility, as follows.
  Because each trade should correspond one-to-one to a feasible prediction, Incentive Compatibility implies there must be no available trades that make a guaranteed profit regardless of the outcome (what the prediction market literature calls an arbitrage opportunity).
  Otherwise, traders with many different beliefs would all choose such a trade.
  In particular, nonnegative trade bundles must cost a positive amount of money, which gives that $C$ is increasing.
  And $\alpha$ units of the grand bundle $(1,\dots,1)$ must cost exactly $\alpha$, since otherwise buying or short-selling it would guarantee profit; this yields ones-invariance.
\end{proof}

\subsection{CFMMs elicit ratios of valuations}
\label{app:cfmms-elicit-ratios}

A consequence of the equivalence is that CFMMs can be viewed as prediction markets.
But what exactly do they ``predict''?
The answer is: \textbf{CFMMs elicit ratios of valuations.}

To explain, let us return to the ``cashless prediction market'', in which there are $n$ assets.
Instead of payments in cash, each asset had a ``price'' in units of the grand bundle $(1,\dots,1)$.
Now generalize: consider a market with $n$ arbitrary assets $A_1,\dots,A_n$ with some nonnegative value.
One can similarly use a cost function $C$ to define a cashless automated market maker with all prices given in units of the grand bundle $(1,\dots,1)$.

When $A_1,\dots,A_n$ are contingent securities (e.g. $A_i$ pays off equal to the number of points scored by team $i$ during the season), this is known as a \emph{ratio-of-expectations} market~\citep{frongillo2018axiomatic}.
The market prices reflect the \emph{ratio} between how the value of $A_1$ and of the grand bundle $A_1 + \cdots + A_n$, because traders pay in units of the latter to obtain units of the former.
However, the market valuation of the grand bundle itself is never revealed.
Similarly, in a CFMM, two traders who value asset $A_i$ at two times asset $A_j$ will accept the same pairwise trades, even if their cash valuations of the grand bundle are very different.

So CFMMs elicit ratios of valuations.
For example, at any history $h$, a CFMM defines via $\valtrades_{\varphi}(h)$ a truthful mechanism to induce a single agent to reveal, not their valuation in cash of the assets, but their ratios of the value of each asset to the value of the grand bundle.
(This is an extension of the fact that CFMMs define proper scoring rules.).

It is not surprising that a CFMM's prices reflect ratios of valuations of the assets.
What may be surprising is that, \emph{if one designed a market maker with the intention of eliciting ratios of valuations}, it would result in a CFMM.
For example, a primary purpose of CFMMs is to ``provide liquidity'', i.e. offer prices to exchange any asset $A_i$ for another $A_j$.
Paying in units of the grand bundle clearly defeats this purpose, which seems to make a ratio-of-valuations market a poor choice.
However, Incentive Compatibility implies that a ratio-of-valuations market should be implemented by a cost-function, which is equivalent to a CFMM that satisfies \liquidation.
So good market-making axioms fall out ``for free'' from elicitation ones.

\section{Cashless prediction markets}
\label{app:cashless}

This section formalizes the equivalence of cost function prediction markets with and without cash.
Recall that, given a cost function $C$, it defines the following automated market maker on the $n+1$ securities $A_1,\dots,A_n,\$$ (Section \ref{sec:define-cost}):
  \[ \valtrades_C(h) = \left\{ (\r,\alpha) \mid \r \in \reals^n \right\} \text{ where $\alpha = C(-\qh - \r) - C(-\qh)$} , \]
Then, we defined the cashless version on $n$ securities $A_1,\dots,A_n$ by Equation \ref{eqn:valtradesprime}, reproduced here:
  \[ \valtrades'_{C}(h) = \left\{ \r + \alpha \ones \mid \r \in \reals^n \right\} \text{ where $\alpha = C(-\qh - \r) - C(-\qh)$} . \]

First, we observe that any valid history and valid trade in the cashless market is a special case of a valid history and trade in the original market
\begin{proposition}
  Consider the cashless version of the market with a valid history $h$, reserves $\qh$, and valid trade $\r \in \valtrades'_C(h)$, resulting in history $h'$ with reserves $\q_{h'} = \qh + \r$.
  Then there is a valid history of the original prediction market reserves $(\qh,0)$ and valid trade $(\r,0)$, resulting in reserves $(\q_{h'},0)$.
  Furthermore, for any realization $Y=y_i$, that trade's net payoff is the same in both markets.
\end{proposition}
\begin{proof}
  Because $\r \in \valtrades'_C(h)$, we can write $\r = \r' + \left(C(-\qh - \r') - C(-\qh)\right)\ones$.
  By ones-invariance, $C(-\qh - \r) = C(-\qh - \r') - C(-\qh - \r') + C(-\qh) = C(-\qh)$.
  So if $(\qh,0)$ is a valid set of reserves in the original market, then $(\r,0)$ is a valid trade in the original market.
  Now by induction on the history, we get that valid histories in the two markets correspond.
  Furthermore, the net payoff to the trader for the trade in either market, under any outcome $Y=y_i$, is $-\r_i$, since the cash transfer is zero in the original market.
\end{proof}

Now, define $f: \reals^{n+1} \to \reals^n$ by $f(\vec{b},a) = \vec{b} + a \ones$.
We can show the more interesting direction of the equivalence as follows:
\begin{proposition}
  Consider a cost-function market with valid history $h$, current reserves $\qh$, and valid trade $(\r,\alpha) \in \valtrades_C(h)$, resulting in history $h'$ with reserves $\q_{h'} = \qh + (-\r,\alpha)$.
  Then in the cashless version, there is a valid history with reserves $f(\qh)$ and valid trade $f(\r,\alpha)$, resulting in a history with reserves $f(\q_{h'})$, such that for any outcome $Y=y_i$, the trade's net payoff is the same as in the original market.
\end{proposition}
\begin{proof}
  If $(\r,\alpha) \in \valtrades_C(h)$, then $\alpha = C(-\qh - \r) - C(-\qh)$.
  Suppose for now $f(\qh)$ are valid reserves in the cashless market.
  Then the trade $f(\r,\alpha) = \r + \alpha \ones$, with $\alpha = C(-\qh-\r) - C(-\qh)$, proving that it is a valid trade in the cashless market.
  Now by induction on the history, we get that valid histories in the two markets correspond.
  Furthermore, the net payoff to the trader for the trade in either market, under any outcome $Y=y_i$, is $-\r_i - \alpha$.
  In the original market, this occurs because the trader pays $\alpha$ units of cash and additionally gives the bundle $\r$, which results in $-\r_i$ when $Y=y_i$.
  In the cashless market, this occurs because the traders gives the bundle $\r + \alpha \ones$, which results in $-\r_i - \alpha$ when $Y=y_i$.
\end{proof}

\end{document}